\documentclass[twocolumn,secnumarabic,amssymb, nobibnotes, aps, prd, superscriptaddress, hidelinks]{revtex4-2}
\usepackage{orcidlink}
\usepackage{amsmath, subfigure}
\usepackage{amssymb, amscd, amsthm, amsfonts}

\newtheorem{theorem}{Theorem}
\newtheorem{lemma}[theorem]{Lemma}
\newtheorem{conjecture}[theorem]{Conjecture}
\newtheorem{proposition}[theorem]{Proposition}

\newtheorem{assumption}{Hypothesis}

\newcommand{\Z}{\mathbb{Z}}

\newcommand{\R}{\mathbb{R}}

\newcommand{\T}{\mathbb{T}}

\newcommand{\D}{\mathrm{d}}
\providecommand{\Zd}{{\mathbb Z^d}}
\providecommand{\nl}{ \nonumber \\ }
\providecommand{\eps}{\varepsilon}

\newcommand{\bn}{{\bf n}}

\newcommand{\inv}{^{-1}}
\newcommand{\veee}[1]{|\!|\!|#1|\!|\!|}

\providecommand{\nnnorm}[1]{\veee #1}
\usepackage{dsfont}
\newcommand{\1}{\mathds{1}}

\providecommand{\abs}[1]{|#1|}

\newcommand{\dc}{ d_{\mathrm{c}} }
\newcommand{\dca}{ d_{\mathrm{c},\alpha} }
\newcommand{\betac}{{\beta_{\mathrm{c}}}}
\newcommand{\nuc}{\nu_{\mathrm{c}}}
\newcommand{\Gbar}{\Gamma}
\newcommand{\free}{\mathrm{F}}
\newcommand{\omegaR}{{\omega_R}}
\newcommand{\bs}{\betac(s)}
\newcommand{\bRs}{\beta_{R,\mathrm{c}}(s)}
\newcommand{\scale}{L}

\usepackage[greek,english]{babel}
\def\coppa{{\fontencoding{LGR}\fontfamily{cmr}\selectfont\textqoppa}}
\def\q{\hbox{\foreignlanguage{greek}{\coppa}}}
\def\qq{{\hbox{\foreignlanguage{greek}{\footnotesize\coppa}}}}

\begin{document}

\title{Universal finite-size scaling in high-dimensional critical phenomena}%

\author{Yucheng Liu\,\orcidlink{0000-0002-1917-8330}}%
\email[Liu: ]{yliu135@pku.edu.cn}
\address{Beijing International Center for Mathematical Research, Peking University, Beijing, China 100871}
\address{Department of Mathematics, University of British Columbia,
Vancouver BC, Canada V6T1Z2}
\author{Jiwoon Park\,\orcidlink{0000-0002-1159-2676}}%
\email[Park: ]{jp711@cantab.ac.uk}
\address{Department of Mathematics,
  Republic of Korea Air Force Academy,
  635, Danjae-ro, Cheongju-si, Chungcheongbuk-do, Republic of Korea}
\author{Gordon Slade\,\orcidlink{0000-0001-9389-9497}}%
\email[Slade: ]{slade@math.ubc.ca}
\address{Department of Mathematics, University of British Columbia,
Vancouver BC, Canada V6T1Z2}

\begin{abstract}
We present a new unified theory of critical finite-size scaling for
lattice statistical mechanical models
with periodic boundary conditions above the upper critical dimension.
Our theory is based on recent mathematically rigorous results for linear and
branched polymers, multi-component spin systems, and percolation.
Both short-range and long-range interactions are included.
The universal finite-size scaling is inherited from the scaling
of the system unwrapped to the infinite lattice.
We also present conjectures for universal scaling profiles for
the susceptibility and two-point function plateau in a critical window.
For free boundary conditions,
the universal scaling
has been proven to apply at a pseudocritical point
for hierarchical spins,
and we conjecture that this holds generally.
\end{abstract}

\maketitle

\section{Introduction}
\label{sec:intro}

The subject of critical phenomena for lattice models in statistical
mechanics is a cornerstone of theoretical physics.  The critical
behaviour of the Ising and multi-component $|\varphi|^4$ models for ferromagnetism,
of percolation, and of linear
and branched polymers, forms a central part of the theory.
The rich and universal fractal geometry connected with
second-order phase transitions
in these models, and its characterisation in terms of critical
exponents, is an ongoing source of fascination both in physics and in mathematics.

It has long been understood that the dependence on the spatial dimension
is an important feature in critical phenomena.
In particular, there is typically an
upper critical dimension $\dc$ above which mean-field critical exponents occur.
For models with short-range (SR) interactions, $\dc=4$ for Ising and $|\varphi|^4$
spin systems, $\dc=4$ for linear polymers (self-avoiding walk),
$\dc=6$ for percolation, and $\dc=8$ for
branched polymers (lattice trees and animals).

Since laboratory samples and simulation experiments involve finite systems,
the finite-size scaling (FSS) of phase transitions forms an important part of
the theory.  Above the upper critical dimension, the role played
by boundary conditions in FSS has been a subject of
some controversy and much discussion, e.g., \cite{LB97,JY05,BKW12,WY14,F-SBKW16,GEZGD17,KB17,ZGFDG18,BEHK22,DGGZ22,FMPPS23,DGGZ24,LM16}.
Some of the history is recounted in \cite{BEHK22}.
It is therefore useful to have mathematically rigorous results for FSS.

In this work, we present a new and unified theory of FSS under
periodic boundary conditions (PBC), based on unwrapping the finite model
to the infinite lattice as in Figure~\ref{figure:LTlift}.
A general comparison principle between the original and the unwrapped
systems is proposed and verified for the Ising model, the self-avoiding walk, percolation, and branched polymers above their upper critical dimensions, using the rigorous results of \cite{Slad23_wsaw,Liu25EJP,HMS23,LPS25-Ising,LS25a}.
Thereby, we obtain mathematically rigorous results on the
finite-size critical exponents for these models under PBC.

\begin{figure}[h]
\parbox{0.4\linewidth}{\includegraphics[scale = 0.6]{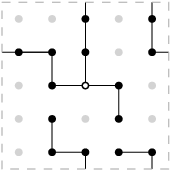}}
\parbox{0.5\linewidth}{\includegraphics[scale = 0.6]{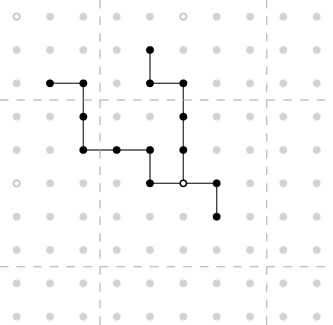}}
\caption{The unwrapping of a lattice tree from the 2-dimensional discrete torus
of period 5 to the infinite lattice.
}
\label{figure:LTlift}
\end{figure}

For all of the SR models mentioned above, we review recent proofs that
in a volume $V$ with PBC
in dimensions $d>\dc$,  and at the infinite-volume critical point,
the susceptibility has size (at least) $V^{\frac{2}{\dc}}$,
and the two-point function has a plateau of size (at least) $V^{\frac{2}{\dc}-1}$.
We also identify a critical window of width (or \emph{rounding scale}) $V^{-\frac{2}{\gamma \dc}}$ around the infinite-volume critical point, where $\gamma$ is the critical exponent for the infinite-volume susceptibility.

Long-range (LR) versions of the models have couplings
decaying like $r^{-(d+\alpha)}$, with $\alpha \in (0,2)$ for $d \ge 2$
and $\alpha \in (0,1)$ for $d=1$ (the enhanced restriction for $d=1$ is
to ensure a second-order phase transition).
LR models have upper critical dimensions $\dca =\frac{\alpha}{2}\dc$, with $\dc$ the
SR upper critical dimension.  For LR models, $d>\dca$
includes low dimensions, even $d=1$, when $\alpha$ is small enough.
Our general theory also applies to LR models above $\dca$.
It is always the short-range $\dc$, not the long-range $\dca$, that appears in the finite-size scaling exponents mentioned in the previous paragraph.
We demonstrate this with the long-range self-avoiding walk, using a recent result of \cite{Liu25}.

Our theory confirms some of the findings of
\cite{LB97,JY05,BKW12,WY14,F-SBKW16,GEZGD17,KB17,ZGFDG18,BEHK22,DGGZ22,DGGZ24,FMPPS23},
via a completely different method.
The mechanisms we reveal
are, to a large extent, model-independent, independent of whether the models
are short- or long-range, and mathematically rigorous.
The relevance of unwrapping has been noted previously in \cite{GEZGD17,DGGZ22},
but our perspective on unwrapping is different.
Standard infinite-volume scaling applies to our unwrapped model and directly produces the finite-size scaling exponents.
The larger-than-system torus correlation length as in,
e.g., \cite{JY05,F-SBKW16,BEHK22,FMPPS23},
belongs to the unwrapped model in our work.
The exponent $\q=d/\dc$ (qoppa) which plays an important role,
e.g., in \cite{F-SBKW16,BEHK22},
emerges in our theory in a guise discussed in
Appendix~\ref{app:coppa}.

Throughout the paper, we discuss short-range and long-range models simultaneously.
Formulas involving $\alpha$ apply to LR models for any choice of $\alpha \in (0,\min\{d,2\})$,
and apply to SR models after setting $\alpha=2$.

\medskip\noindent \emph{Notation.}
For dimension $d\ge 1$,
the hypercubic lattice is denoted $\Z^d$, and the
discrete torus of period $R$ and volume $V=R^d$
is denoted $\T_R^d$.  We sometimes identify a point $x$ in the torus with
its representative in $[ - \frac R 2, \frac R 2 )^d \cap \Z^d$.

We write $f \lesssim g$ or $f = O(g)$ to mean there is a constant $C > 0$ such that $f \le Cg$, write $f \ll g$ to mean $\lim f/g =0$,
and write $f\asymp g$ to mean $f \lesssim g \lesssim f$.
We write $|x|= \max\{|x_1|,\ldots,|x_d|\}$ for the $\ell^\infty$ norm, and $\nnnorm x = \max\{|x| , 1\}$ to avoid division by $0$.

\medskip\noindent \emph{Critical exponents.}
For a given model, let
$\betac$ denote the infinite-volume critical point,
and let $t = (\betac-\beta) / \betac > 0$ denote the reduced inverse temperature.
The two-point functions on $\Z^d$ and $\T_R^d$ are denoted, respectively, by
$G_\beta(x)$ and $G_{\beta,R}(x)$.
Let $\chi(\beta)=\sum_{x\in\Z^d}G_\beta(x)$ denote the susceptibility on $\Zd$, and let
$\chi_R(\beta)=\sum_{x\in\T_R^d}G_{R,\beta}(x)$ denote the susceptibility on $\T_R^d$.
The two-point functions and susceptibilities are not truncated.
We assume that both two-point functions are non-decreasing functions of $\beta$.
Let $\xi(\beta)$ denote the correlation length on $\Zd$.
We do not use a correlation length on the torus.
The $\Zd$ critical exponents $\eta,\gamma,\nu$ are defined by
\begin{gather}
G_{\betac}(x) \asymp \frac 1 {\abs x^{d-2+\eta}}
\qquad (|x|\to\infty),\\
\label{eq:gamma_nu}
\chi(\beta) \asymp t^{-\gamma}, \qquad
\xi(\beta) \asymp t^{-\nu} \qquad (t \downarrow 0).
\end{gather}
Above the upper critical dimension $\dca$,
the critical exponents are proved in many cases
(e.g., \cite{HS92a,HS92c,Saka07,Hara08,FH17,HHS08,CS11,CS15})
to take their mean-field values $\eta = 2 - \alpha$, $\nu = \gamma / \alpha$, and
\begin{equation}
\gamma = \begin{cases}
1		&(\text{self-avoiding walk, spin, percolation}) \\
\frac 12		&(\text{lattice trees and animals}) .
\end{cases}
\end{equation}

\medskip
Our main rigorous result (Theorem~\ref{thm:main} in Section~\ref{sec:plateau-theorem})
gives a sufficient condition for the
short- or long-range torus susceptibility and two-point function to obey
\begin{equation}
\label{eq:PBC_lb}
\chi_R(\betac) \gtrsim V^{\frac 2 {\dc}},
\quad
G_{R,\betac}(x) \gtrsim \frac 1 {\nnnorm x^{d-\alpha}} + \frac 1 {V^{1-\frac 2 {\dc}}}
\end{equation}
uniformly in large $R$, with the short-range $\dc$.
The constant ``plateau'' term in the lower bound of $G_{R,\betac}(x)$ dominates $\abs x^{-(d-\alpha)}$
as soon as
\begin{equation}
\label{eq:profilewins}
    |x| \gg  R^p \;\; \text{with}\;\; p
    = \frac{d-\alpha \frac{d}{\dca}}{d-\alpha} .
\end{equation}
Note that the exponent $p$ is less than $1$ for $d > \dca$, so the plateau dominates
over all but a vanishingly small proportion of the torus.
The plateau term contributes $V^{\frac{ 2  }{ \dc }}
=R^{\frac{ 2 d }{ \dc }} = R^{\frac{\alpha d}{\dca}}$ to the susceptibility,
which is much larger than the contribution $R^\alpha$ from the decaying term in the two-point function.
Our proof also identifies the width of the critical window to be $V^{-\frac{2}{\gamma \dc}}$, and it gives matching upper bounds for \eqref{eq:PBC_lb} at the
edge (high-temperature, disordered side) of the critical window.
Theorem~\ref{thm:main} is proved in Section~\ref{sec:proof},
and in Section~\ref{sec:SAW} we verify its sufficient condition for
the case of LR self-avoiding walk.
Further applications of Theorem~\ref{thm:main} are discussed in Section~\ref{sec:applications}.

Beyond the computation of FSS critical exponents, more precise information about the
scaling of the susceptibility and two-point function plateau can be sought.
For the $n$-component $|\varphi|^4$ model in dimensions $d>4$ with PBC,
an exact profile for the amplitude
of the susceptibility in the window
was computed in 1985 using a Wilsonian renormalisation group method \cite{BZ85,Zinn21}; we review this computation in Appendix~\ref{app:RG}.
A mathematically rigorous derivation of the same profile has been carried out recently on the hierarchical lattice \cite{MPS23,PS25}
for dimensions $d \ge 4$ (including
$d=\dc=4$).
For $n=1$, the profile is the same as its analogue on the complete graph
(Curie--Weiss model).
Despite the results of \cite{BZ85}, the profile has subsequently received scant
attention in the FSS literature, for example it is not explicitly mentioned in any of
\cite{LB97,JY05,BKW12,WY14,F-SBKW16,GEZGD17,KB17,ZGFDG18,BEHK22,DGGZ22,FMPPS23,DGGZ24,LM16}.
Part of our purpose is to focus attention on the importance of the profile.
The profiles for self-avoiding walk and branched polymers on the complete graph
have been computed, and we conjecture that these profiles also apply to SR and LR models at and
above their critical dimensions.
We also discuss a profile for percolation,
which is of a different character than the others.
Our conjectures for the universal profiles are stated in Section~\ref{sec:profile}.

For the hierarchical $n$-component $|\varphi|^4$ model with $d \ge \dc=4$,
it has been proved that, under FBC in volume $V=R^d$,
the universal behaviour observed for PBC in the critical window about the infinite-volume
critical point is exactly duplicated around a pseudocritical point which is shifted
from the infinite-volume critical point by $R^{-1/\nu}=R^{-2}$ for $d>4$
(this shift is also observed, e.g., in \cite{Kenn12,WY14}) and by
$R^{-2}(\log R)^{\frac{n+2}{n+8}}$ for $d=4$ \cite{MPS23,PS25}.
The critical windows for PBC and FBC do not overlap when $d \ge\dc$.
In Section~\ref{sec:fbc}, we
formulate a conjecture that similar behaviour holds more generally.

\section{General plateau theorem}
\label{sec:plateau-theorem}

Our general plateau theorem holds under two hypotheses.
The hypotheses have been verified for many examples, discussed
in Sections~\ref{sec:SAW} and~\ref{sec:applications}.
Our first hypothesis concerns the $\Zd$ two-point function $G_\beta(x)$ near criticality.

\begin{assumption} \label{ass:G}
The susceptibility obeys \eqref{eq:gamma_nu} with
critical exponent $\gamma >0$.
There is a correlation length satisfying $\xi(\beta) \asymp t^{-\nu}$ with
$\nu = \gamma/\alpha$,
a function $g:[0,\infty) \to [0,\infty)$
with $g(u) \lesssim (1+u)^{-(\alpha+\eps)}$ for some $\eps > 0$, and a constant $s_1>0$
such that
\begin{align}
\label{eq:Gdecay}
G_\beta(x) &\lesssim \frac{1}{\nnnorm x^{d-\alpha}} g(|x|/\xi(\beta))
	\qquad(x\in \Zd) ,
\\
\label{eq:Glb}
G_\beta(x) &\gtrsim \frac{1}{\nnnorm x^{d-\alpha}}
	~\quad\quad\qquad( \abs x \le s_1 \xi(\beta)) .
\end{align}
\end{assumption}

Hypothesis~\ref{ass:G} asserts that $G_\beta(x)$ has its critical
decay $\abs{x}^{-(d-\alpha)}$ up to the scale of the correlation length,
beyond which it decays more rapidly; this is essentially a definition of a correlation length.
In applications, the function $g$ takes the form
\begin{equation}
\label{eq:gdef}
g_{\mathrm{SR}}(u) = e^{-cu},
\qquad
g_{\mathrm{LR}}(u) = \frac{1}{(1+cu)^{2\alpha}}
\end{equation}
for SR and LR models respectively.

Our second hypothesis concerns the torus two-point function $G_{R,\beta}(x)$.
For its statement, we introduce
the \emph{unwrapped two-point function}
\begin{equation} \label{eq:Gamma}
    \Gbar_{R,\beta}(x)
    = \sum_{u\in \Z^d} G_\beta (x+Ru)
    \qquad(x\in \T_R^d),
\end{equation}
which can be considered either as a function on the torus or as a periodic
function on $\Z^d$. (Recall we identify $\T_R^d$ with
$[ - \frac R 2, \frac R 2 )^d \cap \Z^d$.)
The sum is over all possible unwrapped locations $x+Ru$ of the torus point $x \in \T_R^d$.
The unwrapped two-point function is inspired by the unwrapping shown in Figure~\ref{figure:LTlift}.

\begin{assumption} \label{ass:Gbar}
There are constants $c_0,s_2>0$ such that
\begin{equation} \label{eq:GGbar}
\Big( 1 -  c_0\frac{\chi(\beta)^{\dc/2}}{V} \Big)
\Gbar_{R,\beta}(x) \le G_{R,\beta}(x) \le \Gbar_{R,\beta}(x),
\end{equation}
uniformly in $x\in \T_R^d$ and in $\beta, R$ satisfying
$\xi(\beta) \ge s_2 R$.
\end{assumption}

Hypothesis~\ref{ass:Gbar}
is a comparison principle between the torus and  unwrapped models.
The upper bound of \eqref{eq:GGbar}
reflects lesser interaction in the unwrapped model, and in applications it
holds in any dimension.
The lower bound of \eqref{eq:GGbar} quantifies the error in the comparison, and it
is expected to hold only above the upper critical dimension $\dca$.
We emphasise it is $\chi^{\dc/2}$, \emph{not} $\chi^{\dca/2}$, that occurs in \eqref{eq:GGbar}, even for LR models.
The power $\dc/2$ arises from the topology of
certain Feynman diagrams, which are identical for SR and LR models (see Appendix~\ref{app:LB}).

\begin{theorem} \label{thm:main}
Let $d > \dca$ and let $R$ be sufficiently large.
Under Hypotheses~\ref{ass:G} and~\ref{ass:Gbar},
there is a constant $c_1>0$ such that, at $\beta_*$
defined by $t_* = (\betac-\beta_*) / \betac  = c_1 V^{-\frac 2 {\gamma \dc}}$,
\begin{equation} \label{eq:PBC}
\chi_R(\beta_*) \asymp V^{\frac 2 {\dc}},
\quad
G_{R, \beta_*}(x) \asymp \frac 1 {\nnnorm x^{d-\alpha}} + \frac 1 {V^{1-\frac 2 {\dc}}}
\end{equation}
uniformly in $R$ and in $x \in \T_R^d$.
\end{theorem}

By the monotonicity in $\beta$, Theorem~\ref{thm:main} immediately implies the lower bounds at $\betac$ claimed in \eqref{eq:PBC_lb}.
Theorem~\ref{thm:main} identifies the exponents mentioned at the
end of Section~\ref{sec:intro}:
the powers $V^{-\frac 2 {\gamma \dc}}, V^{\frac 2 {\dc}}, V^{-(1-\frac 2 {\dc} )}$ are the \emph{window scale, susceptibility scale}, and \emph{plateau scale}, respectively.
See also Table~\ref{table:fss}.

\section{Proof of plateau theorem}
\label{sec:proof}

An ingredient in the proof of Theorem~\ref{thm:main} is
the following lemma, whose elementary proof is given in
Appendix~\ref{app:conv}.

\begin{lemma} \label{lem:periodic_sum}
Let $d \ge 1$, $R \ge 1$,
$a >0$ and $\xi >0$.
Suppose $g:[0,\infty) \to [0,\infty)$ satisfies $g(u) \le (1+u)^{-(a+\eps)}$ for some $\eps > 0$.
Then there is a constant $C = C(d,a,\eps)>0$ such that
\begin{equation*}
\sum_{u \in\Z^d : u \neq 0}
\frac{1}{|x + R u|^{d-a}} g \bigg( \frac{ |x+Ru| } \xi \bigg)
\le C \frac{ \xi^a }{R^d}
\end{equation*}
for all $x\in \Zd$ with $| x| \le R/2$.
\end{lemma}

The next proposition applies Lemma~\ref{lem:periodic_sum} to conclude bounds on the unwrapped two-point function $\Gbar_{R,\beta}(x)$.

\begin{proposition} \label{prop:Gamma}
Under Hypothesis~\ref{ass:G}, and assuming that
$\xi(\beta) \ge \frac{3}{2s_1} R$ for the lower bound in \eqref{eq:Gbar},
we have
\begin{equation} \label{eq:Gbar}
\Gamma_{R,\beta}(x) \asymp \frac 1 { \nnnorm x^{d-\alpha} } + \frac{ \chi(\beta) } V
\end{equation}
uniformly in $x \in \T_R^d$.
\end{proposition}

\begin{proof}
The upper bound is a direct consequence of Lemma~\ref{lem:periodic_sum} (with $a=\alpha$) and Hypothesis~\ref{ass:G} ($\nu=\gamma/\alpha$ implies that $\xi^\alpha \asymp \chi$):
\begin{equation}
\Gamma_{R,\beta}(x)
\le G_\beta(x) + O\Big( \frac{ \xi^\alpha} {R^d} \Big)
\lesssim \frac 1 {\nnnorm x^{d-\alpha} } + \frac \chi V .
\end{equation}

For the lower bound, we restrict the sum in \eqref{eq:Gamma} to $| u| \le M$
with a well-chosen $M \ge 1$.
We then wish to apply \eqref{eq:Glb} to all $G_\beta(x +Ru)$ with $| u| \le M$.
Since $| x| \le R / 2$,
by the triangle inequality we can do so if
\begin{equation}\label{eq:proof1.11}
    M \ge 1,
\qquad
    |x + R u| \le \frac R 2 + R M \le s_1 \xi.
\end{equation}
We choose $M= \frac {2 s_1}{3} \xi / R$
so that both inequalities assert $R \le \frac {2s_1}{ 3}  \xi$, as we have assumed
to hold for the lower bound.
Since $|x+Ru| \le \frac R 2 +R | u| \le \frac 3 2 R | u|$, \eqref{eq:Glb} gives
\begin{align}
\sum_{u\ne 0} G_\beta (x+Ru)
&\gtrsim \sum_{1 \le | u| \le M} \frac 1 { | x+Ru |^{d-\alpha} } \nl
&\ge \sum_{1 \le | u| \le M} \frac 1 { ( \frac 3 2 R | u|  )^{d-\alpha} }
\asymp \frac{M^\alpha}{R^{d-\alpha}} .
\end{align}
Our choice $M= \frac {2s_1}{3} \xi / R$,
the lower bound \eqref{eq:Glb} on $G_\beta(x)$, and $\xi^\alpha \asymp \chi$, then give
\begin{equation}
\Gamma_{R,\beta}(x)
\gtrsim G_\beta(x) + \frac{ \xi^\alpha }{ R^d }
\gtrsim \frac 1 { \nnnorm x^{d-\alpha} } + \frac \chi  V .
\end{equation}
This completes the proof.
\end{proof}

\begin{proof}[Proof of Theorem~\ref{thm:main}]
We start by fixing $c_1$ large enough so that
the subtracted term in the lower bound of Hypothesis~\ref{ass:Gbar} is harmless at $\beta_*$.
Since $\chi(\beta) \asymp t^{-\gamma}$ by \eqref{eq:gamma_nu},
we can and do fix $c_1$ sufficiently large so that
\begin{equation}
c_0 \frac{  \chi(\beta_*)^{\dc/  2 } }V
\le c_0' \frac{  1 } { V t_*^{\gamma \dc / 2 } }
= \frac{  c_0' } { c_1^{\gamma \dc / 2 } }
\le \frac 1 2 .
\end{equation}
Concerning the restrictions $\xi(\beta_*) \ge s_2R$
and $\xi(\beta_*) \ge \frac{3}{2s_1} R$ required by Hypothesis~\ref{ass:Gbar}
and Proposition~\ref{prop:Gamma}, we
use $\xi(\beta) \asymp t^{-\nu} = t^{-\gamma/\alpha}$ by \eqref{eq:gamma_nu} and use $d > \dca$, to see that
for $R$ sufficiently large we have
\begin{equation}
\xi(\beta_*) \asymp \frac 1 {t_*^{\gamma/\alpha}}
= \frac {V^{\frac 2 {\alpha \dc}}} { c_1^{\gamma/\alpha}}
= \frac { R^{\frac d {\dca}}} { c_1^{\gamma/\alpha}}
\gg R .
\end{equation}
Therefore we can apply \eqref{eq:GGbar} and Proposition~\ref{prop:Gamma} at $\beta_*$.
Together, they give
\begin{equation}
G_{R,\beta_*}(x)
\asymp \Gamma_{R,\beta_*}(x)
\asymp \frac 1 { \nnnorm x^{d-\alpha} } + \frac{ \chi(\beta_*) } V
\end{equation}
uniformly in $R$ large and in $x \in \T_R^d$.

The plateau term is
\begin{equation}
\frac{\chi(\beta_*)}{V}
\asymp \frac 1 {V t_*^\gamma }
= \frac {V^{2/\dc}} {V c_1^\gamma }
= \frac 1 {c_1^\gamma V^{1-\frac 2 {\dc}}} ,
\end{equation}
as desired.
Finally, we sum $G_{R,\beta_*}(x)$ over $\T_R^d$ to get
\begin{equation}
\chi_R(\beta_*)
\asymp \sum_{x\in \T_R^d} \Big(
	\frac 1 {\nnnorm x^{d-\alpha}} + \frac 1 {V^{1-\frac 2 {\dc}}} \Big)
\asymp R^\alpha + V^{\frac 2 {\dc}}
\asymp V^{\frac 2 {\dc}} ,
\end{equation}
since $V^{\frac{2}{\dc}} = R^{\frac{2d}{\dc}}$ dominates
$R^\alpha$ when $d > \dca = \frac \alpha 2 \dc$.
This completes the proof.
\end{proof}

\section{Long-range self-avoiding walk}
\label{sec:SAW}

In this section we verify Hypotheses~\ref{ass:G} and~\ref{ass:Gbar}
for the spread-out LR self-avoiding walk (SAW)
with $\alpha\in (0,2)$ in dimensions $d > \dca = 2\alpha$.
This proves that Theorem~\ref{thm:main} applies in this setting,
with $\gamma=1$.

The self-avoiding walk model is defined as follows.
For infinite-volume, we fix a
lattice-symmetric, summable kernel $J: \Zd \to [0,\infty)$,
and we write $J_{x,y} = J(y - x)$ for all $x,y\in \Zd$.
A self-avoiding walk is a finite path
$\omega = (\omega(0),\omega(1),\ldots,\omega(|\omega|) )$ consisting
of $|\omega|$ steps,
with each $\omega(i)$ in $\Z^d$ and with $\omega(i)  \neq \omega(j)$ when $i\neq j$.
The two-point function is defined by
\begin{equation}
\label{eq:GSAW}
    G_\beta(x) =
    \sum_{\omega: 0 \to x}
    \beta^{|\omega|}\prod_{i=1}^{|\omega|}J_{\omega(i-1),\omega(i)},
\end{equation}
where the sum is over all self-avoiding walks from $0$ to $x$ of any length
(including the zero-step walk when $x=0$).
The torus two-point function $G_{R,\beta}(x)$ is defined similarly, with the sum over self-avoiding walks
on the torus and with $J_{x,y}$ replaced by the periodised kernel
\begin{equation} \label{eq:torus_coupling}
    \bar J_{R;x,y} = \sum_{u \in \Z^d} J_{x,y+Ru}
    \qquad (x,y\in \T_R^d).
\end{equation}

The spread-out LR model is defined by choosing $J$ to have the form
\begin{equation}
    J_{0,x} \asymp \frac{1}{L^d}\Big( \frac{1}{1+|x|/L} \Big)^{d+\alpha} ,
\end{equation}
where $L$ is a large parameter.
The reciprocal $L^{-1}$ provides a small parameter which permits expansion
methods to be used, specifically the lace expansion \cite{Slad06}.
In the following, we assume that $L$ is large enough, and verify
Hypotheses~\ref{ass:G} and~\ref{ass:Gbar} for the spread-out long-range SAW.

\begin{proof}[Verification of Hypothesis~\ref{ass:G}.]
It is proved in \cite{HHS08} that $\chi \asymp t^{-\gamma}$
with $\gamma=1$, using the lace expansion.
The critical two-point function is proved in \cite{CS15}, again via the lace expansion, to satisfy
$G_{\beta_c}(x) \asymp \nnnorm x^{-(d-\alpha)}$ (i.e., $\eta=2-\alpha$); Hypothesis~\ref{ass:G}
requires the latter to be enhanced to near-critical $\beta$.
For $\beta\in [\frac 12, \betac]$ and $x\in \Zd$, the near-critical estimate
\begin{equation} \label{eq:SAW_ub}
G_\beta(x) \le \frac {C_L} { \nnnorm x^{d-\alpha} } \frac 1 {(1 + \abs x (\betac - \beta)^{1/\alpha} / L  )^{2\alpha}}
\end{equation}
was proved recently \cite{Liu25}, once more using the lace expansion.
These substantial results establish \eqref{eq:Gdecay} with $g = g_{\mathrm{LR}}$ (with $c=1$)
and $\xi(\beta) = L(\betac-\beta)^{-1/\alpha}$,
which does satisfy $\xi \asymp t^{-\nu}$ with $\nu = 1/\alpha = \gamma / \alpha$.

It remains to verify \eqref{eq:Glb}.
We use the convolution $(f_1*f_2)(x) = \sum_{y \in \Z^d}f_1(x-y)f_2(y)$.
Since $G_\beta(0)=1$, we only need to consider $x\ne 0$.
For this, we begin with the differential inequality
\begin{equation} \label{eq:SAWdi}
\beta \frac \D {\D\beta} G_\beta(x)
\le (G_\beta * G_\beta)(x).
\end{equation}
This is proved by observing that the factor $|\omega|$
brought down by differentiating $\beta^{|\omega|}$ in \eqref{eq:GSAW} may be regarded
as a sum over nonzero points in the walk $\omega$.
Given a point on the walk,
neglecting the self-avoidance between
the parts of the walk before and after that point gives an upper bound, which results
in \eqref{eq:SAWdi} (see \cite[Lemma~1.5.2]{MS93} for details).

It follows from \eqref{eq:SAWdi}, from $G_\beta(x) \lesssim \nnnorm x^{-(d-\alpha)}$,
and from the elementary convolution
estimate Lemma~\ref{lem:G2} (in Appendix~\ref{app:conv})
that
\begin{equation} \label{eq:diff_ineq}
\beta \frac \D {\D\beta} G_\beta(x)
\lesssim \frac 1 { \nnnorm x^{d-2\alpha} } .
\end{equation}
Integration of the above from $\beta$ to $\betac$,
together with the power-law decay of the critical two-point function $G_{\betac}$, gives
\begin{align}
G_\beta(x) &= G_\betac(x) - ( G_\betac(x) - G_\beta(x))
\nonumber \\  &
\ge \frac{a}{\nnnorm{x}^{d-\alpha}} - \frac{A(\betac - \beta)}{\nnnorm{x}^{d-2\alpha}}
\end{align}
for some constants $a,A > 0$.
Since $\xi(\beta) \asymp (\betac - \beta)^{-1 / \alpha}$,
the subtracted term is at most
\begin{equation}
    \frac{1}{\nnnorm{x}^{d-\alpha}} \frac{A'|x|^\alpha}{\xi(\beta)^\alpha}
\end{equation}
with some constant $A'$,
and this is negligible compared to the main term $a|x|^{-(d-\alpha)}$
if $\nnnorm{x} \le \eps \xi(\beta)$ with an $\eps > 0$ sufficiently small.  This verifies \eqref{eq:Glb} with $s_1=\eps$.
\end{proof}

\begin{proof}[Verification of Hypothesis~\ref{ass:Gbar}.]
Let $\pi: \Zd \to \T_R^d$ denote the natural projection, which we also apply to $\omega$ component-wise. We first prove the upper bound, which asserts that
\begin{align}
G_{R,\beta}(x)
    &=\sum_{\omega_R: 0 \to x}
    \beta^{|\omega_R|}\prod_{j=1}^{|\omega_R|}\bar J_{R;\omega_R(j-1),\omega_R(j)}
    \nonumber \\ & \le
    \sum_{ \substack{ x'\in \Z^d \\ \pi(x')=x} }
    \sum_{\omega: 0 \to x'}
    \beta^{|\omega|}\prod_{j=1}^{|\omega|}J_{\omega(j-1),\omega(j)},
\end{align}
where the first sum is over torus SAWs $\omega_R$
and the second line involves a sum over $\Zd$ SAWs $\omega$.

By the definition of the torus kernel $\bar J_R$ in \eqref{eq:torus_coupling},
the torus product can be unwrapped to give
\begin{equation}
\prod_{j=1}^{\abs \omegaR} \bar J_{R; \omegaR(j-1), \omegaR(j)}
= \sum_{\omega} \1_{\{\pi(\omega) = \omegaR\}}
\prod_{j=1}^{\abs \omega} J_{\omega(j-1) \omega(j)} ,
\end{equation}
where the indicator function $\1_{\{\pi(\omega) = \omegaR\}}$ is $1$ if $\omega$ projects
to $\omega_R$ and otherwise is $0$.
We multiply by $\beta^{\abs \omegaR} = \beta^{\abs \omega}$ and sum over $\omegaR$
from $0$ to $x$. The walk $\omega$ must end at some $x'\in \Zd$ with $\pi(x') = x$,
so we get
\begin{equation}
G_{R,\beta}(x)
= \sum_{\substack{x' ,  \,  \omega:0\to x' \\ \pi(x')=x}}
	\1_{\{ \pi(\omega) \text{ is SA}\}}
	\beta^{\abs \omega}
	\prod_{j=1}^{\abs \omega} J_{\omega(j-1) \omega(j)} ,
\end{equation}
where SA denotes self-avoiding.
Since $\pi(\omega)$ is SA implies that $\omega$
is SA, by relaxing the indicator function we get
\begin{align}
\label{eq:SAWR}
G_{R,\beta}(x)
&\le \sum_{\substack{x',\omega:0\to x' \\ \pi(x')=x}}
	\1_{\{ \omega \text{ is SA}\}}
	\beta^{\abs \omega}
	\prod_{j=1}^{\abs \omega} J_{\omega(j-1) \omega(j)} \nl
&= \sum_{ x': \pi(x')=x} G_\beta(x')
= \Gamma_{R,\beta}(x) .
\end{align}
This proves the upper bound of \eqref{eq:GGbar}.

For the lower bound of \eqref{eq:GGbar},
our goal is to find constants $c_0,s_2>0$ such that
if $\xi(\beta) \ge s_2 R$ then
\begin{equation} \label{eq:GGbar2}
\Gbar_{R,\beta}(x)  - G_{R,\beta}(x) \le
 c_0\frac{\chi(\beta)^{\dc/2}}{V}  \Gbar_{R,\beta}(x)
 .
\end{equation}
To begin, we observe from
\eqref{eq:SAWR} that the difference $\Gamma_{R,\beta} - G_{R,\beta}$ arises from self-avoiding walks $\omega$ whose projection $\pi(\omega)$ is not self-avoiding.
Such walks must visit two distinct points $y,y'$ in $\Zd$ with $\pi(y)=\pi(y')$.
Diagrammatically, this is\; \includegraphics[scale = 0.6]{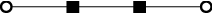}\; where
the hollow circles are $0,x'$ and the filled squares are $y,y'$.
By neglecting the mutual avoidance of the three diagram lines,
the difference $\Gbar_{R,\beta}(x) - G_{R,\beta}(x)$ is therefore bounded above by
\begin{equation} \label{eq:SAW_lb1}
\sum_{y \in \Z^d} G_\beta(y)
\sum_{ y' \neq y: \pi(y')= y,} G_\beta(y'-y)
\sum_{x' :\pi(x')= x}G_\beta(x'-y') .
\end{equation}
The last sum over $x'$ is exactly $\Gamma_{R,\beta}(x - \pi (y))$.
Then we can use \eqref{eq:SAW_ub} and Lemma~\ref{lem:periodic_sum} to bound the middle sum over $y'$ by $O(\xi^\alpha / V)$.
The remaining sum over $y$ is
\begin{equation} \label{eq:SAW_lb2}
\sum_{y\in \Zd}G_\beta(y) \Gamma_{R,\beta}(x-\pi ( y))
= \sum_{u\in \Zd} (G_\beta * G_\beta)(x+Ru).
\end{equation}
Whenever $G_\beta$ satisfies the decay bound \eqref{eq:SAW_ub}, the convolution
can be bounded using Lemma~\ref{lem:G2} in Appendix~\ref{app:conv}, which
gives $G_\beta * G_\beta (y) \lesssim \nnnorm y^{-(d-2\alpha)} (1+|y|/\xi)^{-3\alpha}$.
The $u = 0$ term is then bounded by $\nnnorm x^{-(d-2\alpha)}$.
When $u\neq 0$,  we apply Lemma~\ref{lem:periodic_sum} with $a=2\alpha$ and $\eps = \alpha$, to conclude that the sum over nonzero $u$ is $\lesssim \xi^{2\alpha}/V$.
Altogether, we obtain
\begin{align}
    \Gbar_{R,\beta}(x) - G_{R,\beta}(x)
    \lesssim
    \frac{\xi^\alpha }{V} \Big( \frac{1} { \nnnorm x^{d-2\alpha} } + \frac{ \xi^{2\alpha}} V \Big).
\end{align}
To conclude,
we take $s_2 = \max\{ \frac 1 2 , \frac 3 {2s_1} \}$,
and use $\nnnorm x^\alpha \le (\frac R 2)^\alpha \le (s_2 R)^\alpha \le \xi^\alpha$.
Since also $\xi^\alpha \asymp \chi$, we find that there is a constant $c_0'$ such that
\begin{align}
    \Gbar_{R,\beta}(x) - G_{R,\beta}(x)
    \le
    c_0'
    \frac{\chi^2}{V} \Big( \frac{1} { \nnnorm x^{d-\alpha} } + \frac{ \chi} V \Big).
\end{align}
By Proposition~\ref{prop:Gamma},
the above inequality gives the desired result \eqref{eq:GGbar2},
since $\dc = 4$.
\end{proof}

\section{Further applications of the plateau theorem}
\label{sec:applications}

The plateau theorem has been isolated here, and also applied to long-range
models, for the first time.
However, it has previously been used implicitly in multiple contexts.
We now summarise the relevant literature; the conclusions are
presented in Table~\ref{table:fss}.
(Logarithmic corrections at the critical dimension are discussed in \cite{MPS23,PS25,Kenna04,Ruiz98,Park25b} for SR models.)

\begin{table}[h]
\begin{ruledtabular}
\begin{tabular}{cccc}
& SAW/Spin & Percolation & BP
\\ \hline
$\dc$ & $4$ & $6$ & $8$
\\
$\dca$ & $2\alpha$ & $3\alpha$ & $4\alpha$
\\
$\gamma$ & $1$ & $1$ & $1/2$
\\
window & $V^{-1/2}$ & $V^{-1/3}$ & $V^{-1/2}$
\\
$\chi_R$ & $V^{1/2}$& $V^{1/3}$& $V^{1/4}$
\\
plateau & $V^{-1/2}$  & $V^{-2/3}$ & $V^{-3/4}$
\end{tabular}
\end{ruledtabular}
\caption{Finite-size scaling exponents for short-range models with $d>\dc$ and long-range models with $d>\dca = \frac{\alpha}{2}\dc$.
Rigorous results for long-range models are partial, as discussed in the text.
}
\label{table:fss}
\end{table}

The results concern the following symmetric
kernels $J$ on $\Z^d$
and $\bar J_R$ on $\T_R^d$ for short-range and long-range models:
\begin{itemize}
\item
SR: $J_{x,y}= 1$ for nearest-neighbour $x,y$,
and otherwise $J_{x,y}=0$.
\item
spread-out SR: $J_{x,y}=L^{-d}$ for $0 <  |x-y| \le L$, and otherwise $J_{x,y}=0$
($L$ is fixed and large);
\item
LR: $J_{x,y}\asymp |x-y|^{-(d+\alpha)}$ with $\alpha \in (0,2)$;
\item
spread-out
LR:
$J_{x,y} \asymp L^{-d} (1+|x-y|/L)^{-(d+\alpha)}$
 with $\alpha \in (0,2)$ for $d>1$ and $\alpha\in (0,1)$ for $d=1$
 ($L$ is fixed and large).
\item
torus versions:
$\bar J_{R;x,y} = \sum_{u \in \Z^d} J_{x,y+Ru}$.
\end{itemize}

\smallskip\noindent \textbf{Self-avoiding walk} ($\dc=4$).
The self-avoiding walk on $\Z^d$ has two-point function
defined by \eqref{eq:GSAW}, with any of the above choices of $J$.
The torus two-point function is defined with a sum over torus walks and with
$J$ replaced by $\bar J_R$.
Theorem~\ref{thm:main} has been proved in the following settings, using the lace
expansion:
\begin{itemize}
\item
SR (nearest-neighbour) SAW for $d > 4$ \cite{Slad23_wsaw, Liu25EJP};
\item
spread-out LR SAW for $d > 2\alpha$ in Section~\ref{sec:SAW}.
\end{itemize}

\smallskip\noindent \textbf{Branched Polymer (BP) models} ($\dc=8$).
Let $\mathcal{C}$ denote either a lattice animal (finite connected
bond cluster) or a lattice tree (acyclic animal), and let $|\mathcal{C}|$ denote the number of bonds in $\mathcal{C}$.
The infinite-volume two-point function is defined by
\begin{equation}
    G_\beta(x) = \sum_{\mathcal{C}\ni 0,x} \beta^{|\mathcal{C}|}
    \prod_{\{x,y\} \in \mathcal{C}}J_{x,y},
\end{equation}
with the product over all bonds in $\mathcal{C}$.
The torus two-point function $G_{R,\beta}(x)$ is defined by the sum over $\mathcal{C}$ in the torus, with $J$ replaced by $\bar J_R$.
Theorem~\ref{thm:main} has been proved for:
\begin{itemize}
\item
Spread-out SR BP for $d > 8$ \cite{LS25a}.
\end{itemize}

\smallskip\noindent \textbf{Percolation} ($\dc=6$).
Percolation is a probabilistic model
in which each bond $b$ is \emph{occupied} with probability
$1-e^{-\beta J_b}$ and vacant with probability $e^{-\beta J_b}$.  Bond occupations
are independent.  A configuration of occupied bonds forms
a subgraph of the complete graph on
$\mathbb Z^d$ (or $\T_R^d$), and we say $0 \leftrightarrow x$ if $0$ and $x$ are in the same
connected component of this subgraph.
The two-point functions $G_\beta(x), G_{R,\beta}(x)$ are defined by $\mathrm{Prob}_\beta(0  \leftrightarrow x)$ on $\mathbb Z^d$ or on the torus, respectively.
\begin{itemize}
\item
Theorem~\ref{thm:main} has been proved for
SR percolation for $d\ge 11$, and for spread-out SR percolation for $d > 6$.
Moreover, both asymptotic relations in
\eqref{eq:PBC} have been proved throughout the entire critical window $\abs t \lesssim V^{-1/3}$,
including the critical point and the supercritical regime \cite{HMS23}.
The verification of Hypothesis~\ref{ass:Gbar} is achieved via an exploration
process which couples percolation on $\mathbb Z^d$ and the torus \cite{HHI07}.
\item
For spread-out LR percolation,
the near-critical upper bound \eqref{eq:Gdecay} is proved in
\cite{Liu25} for $\alpha \in (0,2)$.  It is also proved
in \cite{Hutc25} for $\alpha \in (0,1)$ for LR percolation without requiring spread-out.
To extend this to prove Theorem~\ref{thm:main} for LR percolation would require  an
adaptation of the above-mentioned coupling to the LR setting.
\end{itemize}

\smallskip\noindent \textbf{Spin systems} ($\dc=4$).
The Ising model has spins $\varphi_x=\pm 1$
with the Hamiltonian
$H= - \sum_{x \sim y}J_{x,y}\varphi_x \varphi_y$ for $\mathbb Z^d$, and with
instead $\bar J_R$ for the torus.
We also consider the $|\varphi|^4$ model with $n$-component continuous spins $\varphi_x
\in \R^n$;
its Hamiltonian is discussed in Appendix~\ref{app:RG}.
The two-point functions $G_\beta(x), G_{R,\beta}(x)$ are defined by the expectation
$\frac 1n \langle\varphi_0\cdot \varphi_x\rangle_\beta$.
For the Ising model, $\beta$ is the inverse temperature.
For the $|\varphi|^4$ model, the role of $\beta$ is played by $-\nu$ where $\nu$ is
the quadratic coupling constant (see Appendix~\ref{app:RG}).
\begin{itemize}
\item
Theorem~\ref{thm:main} has been proved for the
SR Ising model for $d >4$ \cite{LPS25-Ising}.  The proof relies on the fact
that \eqref{eq:Gdecay}--\eqref{eq:Glb} are proved in \cite{DP25-Ising}.
The verification of Hypothesis~\ref{ass:Gbar}
uses a coupled exploration process and the random current representation
of the Ising model \cite{Aize82}.
\item
For the spread-out LR Ising model,
the upper bound \eqref{eq:Gdecay} is proved in \cite{Liu25}.
To extend this to prove Theorem~\ref{thm:main} would require
adaptation of the above-mentioned coupled exploration process to the LR setting.
\item
For the SR $1$-component $\varphi^4$ model, \eqref{eq:Gdecay}--\eqref{eq:Glb} are proved in \cite{DP25-Ising} for $d>4$.
For the torus,  a rigorous renormalisation group method has recently been used
to prove that for $d >4$ the bounds of Theorem~\ref{thm:main} hold exactly at the infinite-volume critical point for both
$\chi_R (\beta_c)$ and $G_{R,\beta_c}$  \cite{Park25b}, and with logarithmic corrections
for $d=4$.
Stronger results have been proved on the hierarchical lattice for $d\ge4$ \cite{MPS23,PS25}.
We discuss this
in more detail in Section~\ref{sec:profile}.
\end{itemize}

Some ideas behind the proofs of the lower bound of \eqref{eq:GGbar} for
BP, percolation, and the Ising model are sketched in Appendix~\ref{app:LB}.

\section{Universal profiles}
\label{sec:profile}

Theorem~\ref{thm:main} shows that the torus two-point function has a plateau
provided Hypotheses~\ref{ass:G} and~\ref{ass:Gbar} are satisfied.
However, it does not provide precise information on the behaviour of the
two-point function or susceptibility within the critical window.
We now present two conjectures on the behaviour of the torus susceptibility
and plateau throughout the critical window.
We parametrise the window using
\begin{equation}
    \bs = \betac + s a_d V^{-\frac{2}{\gamma\dc}},
\end{equation}
for $s\in \R$, for some non-universal positive constant $a_d$, and with $\betac$
the infinite-volume critical point.
We write $f\sim g$ to mean $\lim f/g = 1$.

\begin{conjecture} \label{conj:profile}
Let $d > \dca$ and $s\in \R$.
For each of the models (SAW, Ising, $|\varphi|^4$, percolation, BP),
there is a model-dependent profile function $f: \R \to (0,\infty)$
(the \emph{same} profile for SR and LR) and positive constants $a_d,b_d$,
such that,  as $R \to \infty$,
\begin{equation} \label{eq:chi_profile}
\chi_R(\bs)
\sim  b_d f(s) V^{\frac 2 {\dc}}.
\end{equation}
Also, as $R \to \infty$,
for every $x\in \T_R^d$,
\begin{equation}
G_{R, \bs}(x) \sim  G_{\betac} (x)
+  \frac{ b_d f(s) }{ V^{1 - \frac 2 {\dc}} } .
\end{equation}
\end{conjecture}

For $k >-1$ and $s\in \R$, let
\begin{equation} \label{eq:Ikdef}
	I_k (s) = \int_{0}^{\infty} x^{k} e^{-\frac 14 x^4  -   \frac 12 s x^2}
    \mathrm d x .
\end{equation}
Also, we define
\begin{equation}
f_{\rm perc}(s)
=   \int_0^\infty
	\frac{\Psi(x^{3/2})}{\sqrt{2\pi x}} e^{-\frac 1 6 x^3 + \frac 1 2 sx^2 - \frac{ 1} 2 s^2 x} \mathrm dx ,
\end{equation}
where $\Psi(z) = \mathbb{E} \exp\{z\int_0^1 W^*(t)\mathrm dt\}$ is the moment generating function of the Brownian excursion area.

\begin{conjecture} \label{conj:profile-f}
The profile functions $f(s)$ are given by the functions in Table~\ref{table:profile},
for both SR and LR models for all $d \ge \dca$ (including the upper critical dimension).
\end{conjecture}

\begin{table}[h]
\begin{ruledtabular}
\begin{tabular}{cccc}
SAW & Spin ($n \ge 1$ components) & Percolation & BP
\\ \hline
$I_1$& $I_{n+1}/(nI_{n-1})$ & $f_{\mathrm{perc}}$ & $I_0$
\end{tabular}
\end{ruledtabular}
\caption{Conjectured universal profile for short-range models with
$d\ge \dc$
and long-range models with $d\ge \dca$.
}
\label{table:profile}
\end{table}

For the $n$-component
$|\varphi|^4$ model with $d>4$,
the profile
\begin{equation}
\label{eq:fndef}
    f_n(s)= \frac{I_{n+1}(s)}{nI_{n-1}(s)}
\end{equation}
is computed in  \cite{BZ85,Zinn21}.
That computation uses a renormalisation group (RG) analysis
which we recall in Appendix~\ref{app:RG}.  A rigorous RG analysis confirms
this profile $f_n$ for the weakly-coupled $n$-component
$|\varphi|^4$ model on a hierarchical lattice \cite{MPS23,PS25},
including for $d=\dc=4$ with logarithmic corrections (see Appendix~\ref{app:coppa}).
More generally, at the upper critical dimension $\dca$, we expect
Conjecture~\ref{conj:profile} to continue to hold with the \emph{same} profile
function $f(s)$, but with logarithmic corrections in the window, susceptibility, and plateau scales.

Despite the fact that $f_n$ was first computed as long ago as 1985,
the existence and form of profile functions in the scaling window have received
scant attention in the finite-size scaling
literature---for example it is not explicitly mentioned in any of
\cite{LB97,JY05,BKW12,WY14,F-SBKW16,GEZGD17,KB17,ZGFDG18,BEHK22,DGGZ22,FMPPS23,DGGZ24,LM16}.
One of our goals in this paper is to draw renewed attention to $f_n$ and other
profiles.

Some of the conjectured profiles have been proven to apply to models defined on the complete
graph ($V$ nodes with edges connecting each pair of nodes).
The $n=1$ profile $f_1$ agrees with the profile for the Curie--Weiss model
(Ising model on complete graph) \cite{EL10}.
The common profile function for lattice trees and lattice animals has been
proven to be $I_0$ on the complete graph
\cite{LS25_profile}.
Integration by parts in the denominator of \eqref{eq:fndef} leads to the formula
\begin{equation}
    f_n(s) = \frac{I_{n+1}(s)}{I_{n+3}(s) + sI_{n+1}(s)} ,
\end{equation}
which extends
the definition of $f_n$ to all $n > -2$.
The $n=0$ case reduces to $f_0(s) = I_1 (s)$ (since
$I_3(x)+sI_1(s)=1$) as in Table~\ref{table:profile};
this is the profile for SAW on the complete graph \cite[Appendix~B]{MPS23}
and is consistent with SAW being the $n=0$ spin model \cite{Genn72}.
The limiting case $f_{-2}(s)=s^{-1}$ is the correct profile for a Gaussian model.
Percolation on the complete graph is the Erd\H{o}s--R\'{e}nyi random graph. 
Its profile function has been proved to be $f_{\mathrm{perc}}$
in \cite[Corollary~3.9]{BRBN26} (see \cite[Appendix~C]{LPS25-Ising} for the identification of the profile function with $f_{\mathrm{perc}}$).
At the moment of final revision of our paper, equation \eqref{eq:chi_profile} in Conjecture~\ref{conj:profile} has been proved for percolation on the high-dimensional torus \cite[Theorem 2.1]{BRN26} for SR spread-out models in dimensions $d >6$, and for
the nearest-neighbour model for $d$ much greater than $6$.

\section{Free boundary conditions}
\label{sec:fbc}

Models with free boundary conditions are formulated in a
discrete $d$-dimensional box $\Lambda_{R,d}$ of side
$R$ and volume $V=R^d$.
They are defined by restricting a kernel $J_{x,y}$ on $\Z^d$ to satisfy
$J_{x,y}=0$ if $x$ or $y$ is not in the box $\Lambda_{R,d}$.
We use the superscript $\free$ to indicate FBC.
Unlike periodic boundary conditions (the torus)
there is no wrapping.

For FBC, the finite-size scaling at the infinite-volume critical point is conjectured---and
in some cases proved---to be different from PBC, namely that
in dimensions $d> \dca$, the susceptibility and two-point function obey
\begin{equation} \label{eq:FBC}
\chi_R^\free(\betac) \asymp R^\alpha,
\qquad
G_{R, \betac}^\free(0,x) \asymp \frac 1 {\abs x^{d-\alpha}} ,
\end{equation}
with $x$ away from the box boundary for
the latter. 
(Since translation invariance is broken, we write both points in $G_{R, \betac}^\free(0,x)$;
the origin is the centre of the box.)
In contrast with \eqref{eq:PBC_lb} for PBC,
there is no plateau in \eqref{eq:FBC}.
This has been proved for the SR Ising model \cite{CJN21} (with the hypothesis supplied by \cite{DP25-Ising,Saka07}) and SR percolation \cite{CH20}.

We conjecture that the $V^{-\frac 2 {\gamma \dc}}, V^{\frac 2 {\dc}}, V^{-(1-\frac 2 {\dc} )}$ scales of window, susceptibility, and plateau, and also the profile functions
of Table~\ref{table:profile}, are \emph{exactly} duplicated at a pseudocritical point $\beta_{R,\mathrm{c}}$ shifted into the ordered phase as
$\beta_{R,\mathrm c} = \betac + \text{const} R^{-1/\nu}$.
(This shift has been observed previously by other authors, e.g., \cite{WY14}.)
Explicitly, with a window around $\beta_{R,\mathrm{c}}$ parametrised by
\begin{equation}
    \bRs = \beta_{R,\mathrm c} + s a_d V^{-\frac{2}{\gamma\dc}},
\end{equation}
we have the following conjecture.  Its constants $a_d$ and $b_d$ are not the same as for PBC.  The restriction $|x|/R\to 0$ for \eqref{eq:FBCprofile} is present to account for
possible $x$ dependence in the profile for $x$ on the scale of $R$.
As noted in \eqref{eq:profilewins}, the profile dominates the Gaussian decay
as soon as $|x| \gg  R^p$ with $p = \frac{d-\alpha \frac{d}{\dca}}{d-\alpha} <1$,
so the restriction $|x|/R\to 0$ still includes a large part of the box where
the universal profile would apply.

\begin{conjecture} \label{conj:FBC}
Let $d > \dca$ and $s\in \R$.
For each of the models (SAW, Ising, $|\varphi|^4$, percolation, BP),
there is a pseudocritical point $\beta_{R,\mathrm c} = \betac + {\rm const}R^{-1/\nu}$ such that, as $R \to \infty$,
\begin{equation}
\chi_R^\free(\bRs )
\sim  b_d f(s) V^{\frac 2 {\dc}}
\end{equation}
with the \emph{same} profile function $f$ from Table~\ref{table:profile}.
Also, as $R \to \infty$, for $x\in \Lambda_{R,d}$ with $|x|/R \to 0$,
\begin{equation}
G^\free_{R, \bRs}(0,x) \sim G_{\betac} (x)
+  \frac{ b_d f(s) }{ V^{1 - \frac 2 {\dc}} } .
	\label{eq:FBCprofile}
\end{equation}
\end{conjecture}

Since $d > \dca = \frac \alpha 2 \dc$,
the shift $R^{-\frac{1}{\nu}} =V^{-\frac{\alpha}{\gamma d}}$
is larger than the window width  $V^{-\frac{2}{\gamma\dc}}$, and the PBC and FBC windows do not overlap.
This shifting of the PBC scaling behaviour to
the FBC pseudocritical point was observed numerically for the 5D SR Ising model
and is discussed in \cite{BEHK22,WY14}
(contrary to the conclusions of \cite{LM16}).

Conjecture~\ref{conj:FBC}, and the Gaussian scaling $\chi_R^\free(\betac) \asymp R^2$ at the infinite-volume critical point,
have been proved for the hierarchical $n$-component
$\abs \varphi^4$ model in \cite{MPS23,PS25}.
This is true also for $d =\dc=4$ with a logarithmic correction to the shift,
which instead of $R^{-1/\nu}=R^{-1/2}$ becomes $R^{-1/2}(\log R)^{\frac{n+2}{n+8}}$.
In general, for all models discussed above, at their upper critical dimension
we expect logarithmic corrections to the window, the susceptibility, and the plateau
to be identical for the (non-overlapping) FBC and PBC windows.

\section{Conclusion}

We have developed a unified theory for finite-size scaling under periodic boundary conditions above the upper critical dimension, and have demonstrated that polymer models, percolation, and spin systems all fit into this general framework.
The theory computes exponents for the scaling window, the critical susceptibility,
and the two-point function plateau.
In our work, a larger-than-system correlation length occurs as the correlation
length of an unwrapped
model.
We propose conjectures for the precise behaviour of the models inside the critical window,
in which the susceptibility and the torus plateau are governed by universal profiles.
These profiles are computed on the complete graph and for a hierarchical model,
and are conjectured to also apply to short- and long-range models with PBC.
Under free boundary conditions, the same behaviour and the same universal profiles
are conjectured to apply in a shifted critical window around a pseudocritical point.

\section*{Acknowledgements}
We thank our collaborators on aspects of this work:
Tom Hutchcroft, Emmanuel Michta, Romain Panis.
The work of YL and GS was supported in part by NSERC of Canada.
YL and GS thank the Isaac Newton Institute for Mathematical Sciences, Cambridge, for support and hospitality during the programme \emph{Stochastic systems for anomalous diffusion}, where work on this paper was undertaken; this work was supported by EPSRC grant EP/R014604/1.  The work of GS was supported in part by the Clay
Mathematics Institute.
JP thanks the Seoul National University for hospitality during
a visit when this work was partially undertaken.
JP was supported by Basic Science Research Program through the National Research Foundation of Korea funded by the Ministry of Science and ICT (RS2025-00518980).


\begin{thebibliography}{55}%
\makeatletter
\providecommand \@ifxundefined [1]{%
 \@ifx{#1\undefined}
}%
\providecommand \@ifnum [1]{%
 \ifnum #1\expandafter \@firstoftwo
 \else \expandafter \@secondoftwo
 \fi
}%
\providecommand \@ifx [1]{%
 \ifx #1\expandafter \@firstoftwo
 \else \expandafter \@secondoftwo
 \fi
}%
\providecommand \natexlab [1]{#1}%
\providecommand \enquote  [1]{``#1''}%
\providecommand \bibnamefont  [1]{#1}%
\providecommand \bibfnamefont [1]{#1}%
\providecommand \citenamefont [1]{#1}%
\providecommand \href@noop [0]{\@secondoftwo}%
\providecommand \href [0]{\begingroup \@sanitize@url \@href}%
\providecommand \@href[1]{\@@startlink{#1}\@@href}%
\providecommand \@@href[1]{\endgroup#1\@@endlink}%
\providecommand \@sanitize@url [0]{\catcode `\\12\catcode `\$12\catcode
  `\&12\catcode `\#12\catcode `\^12\catcode `\_12\catcode `\%12\relax}%
\providecommand \@@startlink[1]{}%
\providecommand \@@endlink[0]{}%
\providecommand \url  [0]{\begingroup\@sanitize@url \@url }%
\providecommand \@url [1]{\endgroup\@href {#1}{\urlprefix }}%
\providecommand \urlprefix  [0]{URL }%
\providecommand \Eprint [0]{\href }%
\providecommand \doibase [0]{https://doi.org/}%
\providecommand \selectlanguage [0]{\@gobble}%
\providecommand \bibinfo  [0]{\@secondoftwo}%
\providecommand \bibfield  [0]{\@secondoftwo}%
\providecommand \translation [1]{[#1]}%
\providecommand \BibitemOpen [0]{}%
\providecommand \bibitemStop [0]{}%
\providecommand \bibitemNoStop [0]{.\EOS\space}%
\providecommand \EOS [0]{\spacefactor3000\relax}%
\providecommand \BibitemShut  [1]{\csname bibitem#1\endcsname}%
\let\auto@bib@innerbib\@empty
\bibitem [{\citenamefont {Luijten}\ and\ \citenamefont
  {Bl\"{o}te}(1997)}]{LB97}%
  \BibitemOpen
  \bibfield  {author} {\bibinfo {author} {\bibfnamefont {E.}~\bibnamefont
  {Luijten}}\ and\ \bibinfo {author} {\bibfnamefont {H.}~\bibnamefont
  {Bl\"{o}te}},\ }\bibfield  {title} {\bibinfo {title} {Classical critical
  behavior of spin models with long-range interactions},\ }\href@noop {}
  {\bibfield  {journal} {\bibinfo  {journal} {Phys. Rev. B}\ }\textbf {\bibinfo
  {volume} {{\bf 56}}},\ \bibinfo {pages} {8945} (\bibinfo {year}
  {1997})}\BibitemShut {NoStop}%
\bibitem [{\citenamefont {Jones}\ and\ \citenamefont {Young}(2005)}]{JY05}%
  \BibitemOpen
  \bibfield  {author} {\bibinfo {author} {\bibfnamefont {J.}~\bibnamefont
  {Jones}}\ and\ \bibinfo {author} {\bibfnamefont {A.}~\bibnamefont {Young}},\
  }\bibfield  {title} {\bibinfo {title} {Finite-size scaling of the correlation
  length above the upper critical dimension in the five-dimensional {Ising}
  model},\ }\href@noop {} {\bibfield  {journal} {\bibinfo  {journal} {Phys.
  Rev. B}\ }\textbf {\bibinfo {volume} {{\bf 71}}},\ \bibinfo {pages} {174438}
  (\bibinfo {year} {2005})}\BibitemShut {NoStop}%
\bibitem [{\citenamefont {Berche}\ \emph {et~al.}(2012)\citenamefont {Berche},
  \citenamefont {Kenna},\ and\ \citenamefont {Walter}}]{BKW12}%
  \BibitemOpen
  \bibfield  {author} {\bibinfo {author} {\bibfnamefont {B.}~\bibnamefont
  {Berche}}, \bibinfo {author} {\bibfnamefont {R.}~\bibnamefont {Kenna}},\ and\
  \bibinfo {author} {\bibfnamefont {J.-C.}\ \bibnamefont {Walter}},\ }\bibfield
   {title} {\bibinfo {title} {Hyperscaling above the upper critical
  dimension},\ }\href@noop {} {\bibfield  {journal} {\bibinfo  {journal} {Nucl.
  Phys. B}\ }\textbf {\bibinfo {volume} {{\bf 865} [FS]}},\ \bibinfo {pages}
  {115} (\bibinfo {year} {2012})}\BibitemShut {NoStop}%
\bibitem [{\citenamefont {Wittmann}\ and\ \citenamefont {Young}(2014)}]{WY14}%
  \BibitemOpen
  \bibfield  {author} {\bibinfo {author} {\bibfnamefont {M.}~\bibnamefont
  {Wittmann}}\ and\ \bibinfo {author} {\bibfnamefont {A.}~\bibnamefont
  {Young}},\ }\bibfield  {title} {\bibinfo {title} {Finite-size scaling above
  the upper critical dimension},\ }\href@noop {} {\bibfield  {journal}
  {\bibinfo  {journal} {Phys. Rev. E}\ }\textbf {\bibinfo {volume} {{\bf
  90}}},\ \bibinfo {pages} {062137} (\bibinfo {year} {2014})}\BibitemShut
  {NoStop}%
\bibitem [{\citenamefont {Flores-Sola}\ \emph {et~al.}(2016)\citenamefont
  {Flores-Sola}, \citenamefont {Berche}, \citenamefont {Kenna},\ and\
  \citenamefont {Weigel}}]{F-SBKW16}%
  \BibitemOpen
  \bibfield  {author} {\bibinfo {author} {\bibfnamefont {E.}~\bibnamefont
  {Flores-Sola}}, \bibinfo {author} {\bibfnamefont {B.}~\bibnamefont {Berche}},
  \bibinfo {author} {\bibfnamefont {R.}~\bibnamefont {Kenna}},\ and\ \bibinfo
  {author} {\bibfnamefont {M.}~\bibnamefont {Weigel}},\ }\bibfield  {title}
  {\bibinfo {title} {Role of {Fourier} modes in finite-size scaling above the
  upper critical dimension},\ }\href@noop {} {\bibfield  {journal} {\bibinfo
  {journal} {Phys. Rev. Lett.}\ }\textbf {\bibinfo {volume} {{\bf 116}}},\
  \bibinfo {pages} {115701} (\bibinfo {year} {2016})}\BibitemShut {NoStop}%
\bibitem [{\citenamefont {Grimm}\ \emph {et~al.}(2017)\citenamefont {Grimm},
  \citenamefont {El\c{c}i}, \citenamefont {Zhou}, \citenamefont {Garoni},\ and\
  \citenamefont {Deng}}]{GEZGD17}%
  \BibitemOpen
  \bibfield  {author} {\bibinfo {author} {\bibfnamefont {J.}~\bibnamefont
  {Grimm}}, \bibinfo {author} {\bibfnamefont {E.}~\bibnamefont {El\c{c}i}},
  \bibinfo {author} {\bibfnamefont {Z.}~\bibnamefont {Zhou}}, \bibinfo {author}
  {\bibfnamefont {T.}~\bibnamefont {Garoni}},\ and\ \bibinfo {author}
  {\bibfnamefont {Y.}~\bibnamefont {Deng}},\ }\bibfield  {title} {\bibinfo
  {title} {Geometric explanation of anomalous finite-size scaling in high
  dimensions},\ }\href@noop {} {\bibfield  {journal} {\bibinfo  {journal}
  {Phys. Rev. Lett.}\ }\textbf {\bibinfo {volume} {{\bf 118}}},\ \bibinfo
  {pages} {115701} (\bibinfo {year} {2017})}\BibitemShut {NoStop}%
\bibitem [{\citenamefont {Kenna}\ and\ \citenamefont {Berche}(2017)}]{KB17}%
  \BibitemOpen
  \bibfield  {author} {\bibinfo {author} {\bibfnamefont {R.}~\bibnamefont
  {Kenna}}\ and\ \bibinfo {author} {\bibfnamefont {B.}~\bibnamefont {Berche}},\
  }\bibfield  {title} {\bibinfo {title} {Universal finite-size scaling for
  percolation theory in high dimensions},\ }\href@noop {} {\bibfield  {journal}
  {\bibinfo  {journal} {J. Phys. A: Math. Theor.}\ }\textbf {\bibinfo {volume}
  {{\bf 50}}},\ \bibinfo {pages} {235001} (\bibinfo {year}
  {2017})}\BibitemShut {NoStop}%
\bibitem [{\citenamefont {Zhou}\ \emph {et~al.}(2018)\citenamefont {Zhou},
  \citenamefont {Grimm}, \citenamefont {Fang}, \citenamefont {Deng},\ and\
  \citenamefont {Garoni}}]{ZGFDG18}%
  \BibitemOpen
  \bibfield  {author} {\bibinfo {author} {\bibfnamefont {Z.}~\bibnamefont
  {Zhou}}, \bibinfo {author} {\bibfnamefont {J.}~\bibnamefont {Grimm}},
  \bibinfo {author} {\bibfnamefont {S.}~\bibnamefont {Fang}}, \bibinfo {author}
  {\bibfnamefont {Y.}~\bibnamefont {Deng}},\ and\ \bibinfo {author}
  {\bibfnamefont {T.}~\bibnamefont {Garoni}},\ }\bibfield  {title} {\bibinfo
  {title} {Random-length random walks and finite-size scaling in high
  dimensions},\ }\href@noop {} {\bibfield  {journal} {\bibinfo  {journal}
  {Phys. Rev. Lett.}\ }\textbf {\bibinfo {volume} {{\bf 121}}},\ \bibinfo
  {pages} {185701} (\bibinfo {year} {2018})}\BibitemShut {NoStop}%
\bibitem [{\citenamefont {Berche}\ \emph {et~al.}(2022)\citenamefont {Berche},
  \citenamefont {Ellis}, \citenamefont {Holovatch},\ and\ \citenamefont
  {Kenna}}]{BEHK22}%
  \BibitemOpen
  \bibfield  {author} {\bibinfo {author} {\bibfnamefont {B.}~\bibnamefont
  {Berche}}, \bibinfo {author} {\bibfnamefont {T.}~\bibnamefont {Ellis}},
  \bibinfo {author} {\bibfnamefont {Y.}~\bibnamefont {Holovatch}},\ and\
  \bibinfo {author} {\bibfnamefont {R.}~\bibnamefont {Kenna}},\ }\bibfield
  {title} {\bibinfo {title} {Phase transitions above the upper critical
  dimension},\ }\href@noop {} {\bibfield  {journal} {\bibinfo  {journal}
  {SciPost Phys. Lect. Notes}\ }\textbf {\bibinfo {volume} {{\bf 136}}},\
  \bibinfo {pages} {paper 60} (\bibinfo {year} {2022})}\BibitemShut {NoStop}%
\bibitem [{\citenamefont {Deng}\ \emph {et~al.}(2022)\citenamefont {Deng},
  \citenamefont {Garoni}, \citenamefont {Grimm},\ and\ \citenamefont
  {Zhou}}]{DGGZ22}%
  \BibitemOpen
  \bibfield  {author} {\bibinfo {author} {\bibfnamefont {Y.}~\bibnamefont
  {Deng}}, \bibinfo {author} {\bibfnamefont {T.}~\bibnamefont {Garoni}},
  \bibinfo {author} {\bibfnamefont {J.}~\bibnamefont {Grimm}},\ and\ \bibinfo
  {author} {\bibfnamefont {Z.}~\bibnamefont {Zhou}},\ }\bibfield  {title}
  {\bibinfo {title} {Unwrapped two-point functions on high-dimensional tori},\
  }\href@noop {} {\bibfield  {journal} {\bibinfo  {journal} {J. Stat. Mech:
  Theory Exp.}\ }\textbf {\bibinfo {volume} {053208}} (\bibinfo {year}
  {2022})}\BibitemShut {NoStop}%
\bibitem [{\citenamefont {Fytas}\ \emph {et~al.}(2023)\citenamefont {Fytas},
  \citenamefont {Mart\'{i}n-Mayor}, \citenamefont {Parisi}, \citenamefont
  {Picco},\ and\ \citenamefont {Sourlas}}]{FMPPS23}%
  \BibitemOpen
  \bibfield  {author} {\bibinfo {author} {\bibfnamefont {N.}~\bibnamefont
  {Fytas}}, \bibinfo {author} {\bibfnamefont {V.}~\bibnamefont
  {Mart\'{i}n-Mayor}}, \bibinfo {author} {\bibfnamefont {G.}~\bibnamefont
  {Parisi}}, \bibinfo {author} {\bibfnamefont {M.}~\bibnamefont {Picco}},\ and\
  \bibinfo {author} {\bibfnamefont {N.}~\bibnamefont {Sourlas}},\ }\bibfield
  {title} {\bibinfo {title} {Finite-size scaling of the random-field {Ising}
  model above the upper critical dimension},\ }\href@noop {} {\bibfield
  {journal} {\bibinfo  {journal} {Phys. Rev. E}\ }\textbf {\bibinfo {volume}
  {{\bf 108}}},\ \bibinfo {pages} {044146} (\bibinfo {year}
  {2023})}\BibitemShut {NoStop}%
\bibitem [{\citenamefont {Deng}\ \emph {et~al.}(2024)\citenamefont {Deng},
  \citenamefont {Garoni}, \citenamefont {Grimm},\ and\ \citenamefont
  {Zhou}}]{DGGZ24}%
  \BibitemOpen
  \bibfield  {author} {\bibinfo {author} {\bibfnamefont {Y.}~\bibnamefont
  {Deng}}, \bibinfo {author} {\bibfnamefont {T.}~\bibnamefont {Garoni}},
  \bibinfo {author} {\bibfnamefont {J.}~\bibnamefont {Grimm}},\ and\ \bibinfo
  {author} {\bibfnamefont {Z.}~\bibnamefont {Zhou}},\ }\bibfield  {title}
  {\bibinfo {title} {Two-point functions of random-length random walk on
  high-dimensional boxes},\ }\href@noop {} {\bibfield  {journal} {\bibinfo
  {journal} {J. Stat. Mech: Theory Exp.}\ }\textbf {\bibinfo {volume} {023203}}
  (\bibinfo {year} {2024})}\BibitemShut {NoStop}%
\bibitem [{\citenamefont {Lundow}\ and\ \citenamefont
  {Markstr\"om}(2016)}]{LM16}%
  \BibitemOpen
  \bibfield  {author} {\bibinfo {author} {\bibfnamefont {P.}~\bibnamefont
  {Lundow}}\ and\ \bibinfo {author} {\bibfnamefont {K.}~\bibnamefont
  {Markstr\"om}},\ }\bibfield  {title} {\bibinfo {title} {The scaling window of
  the $5${D} {Ising} model with free boundary conditions},\ }\href@noop {}
  {\bibfield  {journal} {\bibinfo  {journal} {Nucl. Phys. B}\ }\textbf
  {\bibinfo {volume} {{\bf 911}}},\ \bibinfo {pages} {163} (\bibinfo {year}
  {2016})}\BibitemShut {NoStop}%
\bibitem [{\citenamefont {Slade}(2023)}]{Slad23_wsaw}%
  \BibitemOpen
  \bibfield  {author} {\bibinfo {author} {\bibfnamefont {G.}~\bibnamefont
  {Slade}},\ }\bibfield  {title} {\bibinfo {title} {The near-critical two-point
  function and the torus plateau for weakly self-avoiding walk in high
  dimensions},\ }\href@noop {} {\bibfield  {journal} {\bibinfo  {journal}
  {Math. Phys. Anal. Geom.}\ }\textbf {\bibinfo {volume} {{\bf 26}}},\ \bibinfo
  {pages} {article 6} (\bibinfo {year} {2023})}\BibitemShut {NoStop}%
\bibitem [{\citenamefont {Liu}(2025{\natexlab{a}})}]{Liu25EJP}%
  \BibitemOpen
  \bibfield  {author} {\bibinfo {author} {\bibfnamefont {Y.}~\bibnamefont
  {Liu}},\ }\bibfield  {title} {\bibinfo {title} {A general approach to massive
  upper bound for two-point function with application to self-avoiding walk
  torus plateau},\ }\href@noop {} {\bibfield  {journal} {\bibinfo  {journal}
  {Electr. J. Probab.}\ }\textbf {\bibinfo {volume} {{\bf 30}}},\ \bibinfo
  {pages} {article 110} (\bibinfo {year} {2025}{\natexlab{a}})}\BibitemShut
  {NoStop}%
\bibitem [{\citenamefont {Hutchcroft}\ \emph {et~al.}(2023)\citenamefont
  {Hutchcroft}, \citenamefont {Michta},\ and\ \citenamefont {Slade}}]{HMS23}%
  \BibitemOpen
  \bibfield  {author} {\bibinfo {author} {\bibfnamefont {T.}~\bibnamefont
  {Hutchcroft}}, \bibinfo {author} {\bibfnamefont {E.}~\bibnamefont {Michta}},\
  and\ \bibinfo {author} {\bibfnamefont {G.}~\bibnamefont {Slade}},\ }\bibfield
   {title} {\bibinfo {title} {High-dimensional near-critical percolation and
  the torus plateau},\ }\href@noop {} {\bibfield  {journal} {\bibinfo
  {journal} {Ann. Probab.}\ }\textbf {\bibinfo {volume} {{\bf 51}}},\ \bibinfo
  {pages} {580} (\bibinfo {year} {2023})}\BibitemShut {NoStop}%
\bibitem [{\citenamefont {Liu}\ \emph {et~al.}(2025)\citenamefont {Liu},
  \citenamefont {Panis},\ and\ \citenamefont {Slade}}]{LPS25-Ising}%
  \BibitemOpen
  \bibfield  {author} {\bibinfo {author} {\bibfnamefont {Y.}~\bibnamefont
  {Liu}}, \bibinfo {author} {\bibfnamefont {R.}~\bibnamefont {Panis}},\ and\
  \bibinfo {author} {\bibfnamefont {G.}~\bibnamefont {Slade}},\ }\bibfield
  {title} {\bibinfo {title} {The torus plateau for the high-dimensional {Ising}
  model},\ }\href@noop {} {\bibfield  {journal} {\bibinfo  {journal} {Commun.
  Math. Phys.}\ }\textbf {\bibinfo {volume} {{\bf 406}}},\ \bibinfo {pages}
  {article 159} (\bibinfo {year} {2025})}\BibitemShut {NoStop}%
\bibitem [{\citenamefont {Liu}\ and\ \citenamefont
  {Slade}(2025{\natexlab{a}})}]{LS25a}%
  \BibitemOpen
  \bibfield  {author} {\bibinfo {author} {\bibfnamefont {Y.}~\bibnamefont
  {Liu}}\ and\ \bibinfo {author} {\bibfnamefont {G.}~\bibnamefont {Slade}},\
  }\bibfield  {title} {\bibinfo {title} {Near-critical and finite-size scaling
  for high-dimensional lattice trees and animals},\ }\href@noop {} {\bibfield
  {journal} {\bibinfo  {journal} {J. Stat. Phys.}\ }\textbf {\bibinfo {volume}
  {{\bf 192}}},\ \bibinfo {pages} {article 23} (\bibinfo {year}
  {2025}{\natexlab{a}})}\BibitemShut {NoStop}%
\bibitem [{\citenamefont {Liu}(2025{\natexlab{b}})}]{Liu25}%
  \BibitemOpen
  \bibfield  {author} {\bibinfo {author} {\bibfnamefont {Y.}~\bibnamefont
  {Liu}},\ }\bibfield  {title} {\bibinfo {title} {High-dimensional long-range
  statistical mechanical models have random walk correlation functions},\
  }\href@noop {} {\bibfield  {journal} {\bibinfo  {journal} {Electr. J.
  Probab.}\ }\textbf {\bibinfo {volume} {{\bf 30}}},\ \bibinfo {pages} {article
  192} (\bibinfo {year} {2025}{\natexlab{b}})}\BibitemShut {NoStop}%
\bibitem [{\citenamefont {Hara}\ and\ \citenamefont
  {Slade}(1992{\natexlab{a}})}]{HS92a}%
  \BibitemOpen
  \bibfield  {author} {\bibinfo {author} {\bibfnamefont {T.}~\bibnamefont
  {Hara}}\ and\ \bibinfo {author} {\bibfnamefont {G.}~\bibnamefont {Slade}},\
  }\bibfield  {title} {\bibinfo {title} {Self-avoiding walk in five or more
  dimensions. {I.} {The} critical behaviour},\ }\href@noop {} {\bibfield
  {journal} {\bibinfo  {journal} {Commun.\ Math.\ Phys.}\ }\textbf {\bibinfo
  {volume} {{\bf 147}}},\ \bibinfo {pages} {101} (\bibinfo {year}
  {1992}{\natexlab{a}})}\BibitemShut {NoStop}%
\bibitem [{\citenamefont {Hara}\ and\ \citenamefont
  {Slade}(1992{\natexlab{b}})}]{HS92c}%
  \BibitemOpen
  \bibfield  {author} {\bibinfo {author} {\bibfnamefont {T.}~\bibnamefont
  {Hara}}\ and\ \bibinfo {author} {\bibfnamefont {G.}~\bibnamefont {Slade}},\
  }\bibfield  {title} {\bibinfo {title} {The number and size of branched
  polymers in high dimensions},\ }\href@noop {} {\bibfield  {journal} {\bibinfo
   {journal} {J. Stat. Phys.}\ }\textbf {\bibinfo {volume} {{\bf 67}}},\
  \bibinfo {pages} {1009} (\bibinfo {year} {1992}{\natexlab{b}})}\BibitemShut
  {NoStop}%
\bibitem [{\citenamefont {Sakai}(2007)}]{Saka07}%
  \BibitemOpen
  \bibfield  {author} {\bibinfo {author} {\bibfnamefont {A.}~\bibnamefont
  {Sakai}},\ }\bibfield  {title} {\bibinfo {title} {Lace expansion for the
  {Ising} model},\ }\href@noop {} {\bibfield  {journal} {\bibinfo  {journal}
  {Commun. Math. Phys.}\ }\textbf {\bibinfo {volume} {{\bf 272}}},\ \bibinfo
  {pages} {283} (\bibinfo {year} {2007})},\ \bibinfo {note} {correction:
  A.~Sakai. Correct bounds on the Ising lace-expansion coefficients. {\it
  Commun. Math. Phys.}, {\bf 392}:783--823, (2022).}\BibitemShut {Stop}%
\bibitem [{\citenamefont {Hara}(2008)}]{Hara08}%
  \BibitemOpen
  \bibfield  {author} {\bibinfo {author} {\bibfnamefont {T.}~\bibnamefont
  {Hara}},\ }\bibfield  {title} {\bibinfo {title} {Decay of correlations in
  nearest-neighbor self-avoiding walk, percolation, lattice trees and
  animals},\ }\href@noop {} {\bibfield  {journal} {\bibinfo  {journal} {Ann.
  Probab.}\ }\textbf {\bibinfo {volume} {{\bf 36}}},\ \bibinfo {pages} {530}
  (\bibinfo {year} {2008})}\BibitemShut {NoStop}%
\bibitem [{\citenamefont {Fitzner}\ and\ \citenamefont {Hofstad}(2017)}]{FH17}%
  \BibitemOpen
  \bibfield  {author} {\bibinfo {author} {\bibfnamefont {R.}~\bibnamefont
  {Fitzner}}\ and\ \bibinfo {author} {\bibfnamefont {R.~v.~d.}\ \bibnamefont
  {Hofstad}},\ }\bibfield  {title} {\bibinfo {title} {Mean-field behavior for
  nearest-neighbor percolation in $d>10$},\ }\href@noop {} {\bibfield
  {journal} {\bibinfo  {journal} {Electron. J. Probab.}\ }\textbf {\bibinfo
  {volume} {{\bf 22}}},\ \bibinfo {pages} {1} (\bibinfo {year}
  {2017})}\BibitemShut {NoStop}%
\bibitem [{\citenamefont {Heydenreich}\ \emph {et~al.}(2008)\citenamefont
  {Heydenreich}, \citenamefont {Hofstad},\ and\ \citenamefont {Sakai}}]{HHS08}%
  \BibitemOpen
  \bibfield  {author} {\bibinfo {author} {\bibfnamefont {M.}~\bibnamefont
  {Heydenreich}}, \bibinfo {author} {\bibfnamefont {R.~v.~d.}\ \bibnamefont
  {Hofstad}},\ and\ \bibinfo {author} {\bibfnamefont {A.}~\bibnamefont
  {Sakai}},\ }\bibfield  {title} {\bibinfo {title} {Mean-field behavior for
  long- and finite range {Ising} model, percolation and self-avoiding walk},\
  }\href@noop {} {\bibfield  {journal} {\bibinfo  {journal} {J. Stat. Phys.}\
  }\textbf {\bibinfo {volume} {{\bf 132}}},\ \bibinfo {pages} {1001} (\bibinfo
  {year} {2008})}\BibitemShut {NoStop}%
\bibitem [{\citenamefont {Chen}\ and\ \citenamefont {Sakai}(2011)}]{CS11}%
  \BibitemOpen
  \bibfield  {author} {\bibinfo {author} {\bibfnamefont {L.-C.}\ \bibnamefont
  {Chen}}\ and\ \bibinfo {author} {\bibfnamefont {A.}~\bibnamefont {Sakai}},\
  }\bibfield  {title} {\bibinfo {title} {Asymptotic behavior of the gyration
  radius for long-range self-avoiding walk and long-range oriented
  percolation},\ }\href@noop {} {\bibfield  {journal} {\bibinfo  {journal}
  {Ann.\ Probab.}\ }\textbf {\bibinfo {volume} {{\bf 39}}},\ \bibinfo {pages}
  {507} (\bibinfo {year} {2011})}\BibitemShut {NoStop}%
\bibitem [{\citenamefont {Chen}\ and\ \citenamefont {Sakai}(2015)}]{CS15}%
  \BibitemOpen
  \bibfield  {author} {\bibinfo {author} {\bibfnamefont {L.-C.}\ \bibnamefont
  {Chen}}\ and\ \bibinfo {author} {\bibfnamefont {A.}~\bibnamefont {Sakai}},\
  }\bibfield  {title} {\bibinfo {title} {Critical two-point functions for
  long-range statistical-mechanical models in high dimensions},\ }\href@noop {}
  {\bibfield  {journal} {\bibinfo  {journal} {Ann.\ Probab.}\ }\textbf
  {\bibinfo {volume} {{\bf 43}}},\ \bibinfo {pages} {639} (\bibinfo {year}
  {2015})}\BibitemShut {NoStop}%
\bibitem [{\citenamefont {Br\'ezin}\ and\ \citenamefont
  {Zinn-Justin}(1985)}]{BZ85}%
  \BibitemOpen
  \bibfield  {author} {\bibinfo {author} {\bibfnamefont {E.}~\bibnamefont
  {Br\'ezin}}\ and\ \bibinfo {author} {\bibfnamefont {J.}~\bibnamefont
  {Zinn-Justin}},\ }\bibfield  {title} {\bibinfo {title} {Finite size effects
  in phase transitions},\ }\href@noop {} {\bibfield  {journal} {\bibinfo
  {journal} {Nucl. Phys. B}\ }\textbf {\bibinfo {volume} {{\bf 257}}},\
  \bibinfo {pages} {867} (\bibinfo {year} {1985})}\BibitemShut {NoStop}%
\bibitem [{\citenamefont {Zinn-Justin}(2021)}]{Zinn21}%
  \BibitemOpen
  \bibfield  {author} {\bibinfo {author} {\bibfnamefont {J.}~\bibnamefont
  {Zinn-Justin}},\ }\href@noop {} {\emph {\bibinfo {title} {Quantum Field
  Theory and Critical Phenomena}}},\ \bibinfo {edition} {5th}\ ed.\ (\bibinfo
  {publisher} {Oxford University Press},\ \bibinfo {address} {Oxford},\
  \bibinfo {year} {2021})\BibitemShut {NoStop}%
\bibitem [{\citenamefont {Michta}\ \emph {et~al.}(2025)\citenamefont {Michta},
  \citenamefont {Park},\ and\ \citenamefont {Slade}}]{MPS23}%
  \BibitemOpen
  \bibfield  {author} {\bibinfo {author} {\bibfnamefont {E.}~\bibnamefont
  {Michta}}, \bibinfo {author} {\bibfnamefont {J.}~\bibnamefont {Park}},\ and\
  \bibinfo {author} {\bibfnamefont {G.}~\bibnamefont {Slade}},\ }\bibfield
  {title} {\bibinfo {title} {Boundary conditions and universal finite-size
  scaling for the hierarchical $|\varphi|^4$ model in dimensions $4$ and
  higher},\ }\href@noop {} {\bibfield  {journal} {\bibinfo  {journal} {Commun.
  Pure Appl. Math.}\ }\textbf {\bibinfo {volume} {{\bf 78}}},\ \bibinfo {pages}
  {2001} (\bibinfo {year} {2025})}\BibitemShut {NoStop}%
\bibitem [{\citenamefont {Park}\ and\ \citenamefont {Slade}(2025)}]{PS25}%
  \BibitemOpen
  \bibfield  {author} {\bibinfo {author} {\bibfnamefont {J.}~\bibnamefont
  {Park}}\ and\ \bibinfo {author} {\bibfnamefont {G.}~\bibnamefont {Slade}},\
  }\bibfield  {title} {\bibinfo {title} {Boundary conditions and the two-point
  function plateau for the hierarchical $|\varphi|^4$ model in dimensions $4$
  and higher},\ }\href@noop {} {\bibfield  {journal} {\bibinfo  {journal} {Ann.
  Henri Poincar\'{e}}\ } (\bibinfo {year} {2025})},\ \bibinfo {note}
  {published online.
  \url{https://doi.org/10.1007/s00023-025-01566-y}}\BibitemShut {NoStop}%
\bibitem [{\citenamefont {Kenna}(2012)}]{Kenn12}%
  \BibitemOpen
  \bibfield  {author} {\bibinfo {author} {\bibfnamefont {R.}~\bibnamefont
  {Kenna}},\ }\bibfield  {title} {\bibinfo {title} {Universal scaling relations
  for logarithmic-correction exponents},\ }in\ \href@noop {} {\emph {\bibinfo
  {booktitle} {Order, Disorder, and Criticality: Advanced Problems of Phase
  Transition Theory, {\em {Volume} 3}}}},\ \bibinfo {editor} {edited by\
  \bibinfo {editor} {\bibfnamefont {Y.}~\bibnamefont {Holovatch}}}\ (\bibinfo
  {publisher} {World Scientific},\ \bibinfo {address} {Singapore},\ \bibinfo
  {year} {2012})\ pp.\ \bibinfo {pages} {1--46}\BibitemShut {NoStop}%
\bibitem [{\citenamefont {Slade}(2006)}]{Slad06}%
  \BibitemOpen
  \bibfield  {author} {\bibinfo {author} {\bibfnamefont {G.}~\bibnamefont
  {Slade}},\ }\href@noop {} {\emph {\bibinfo {title} {The Lace Expansion and
  its Applications.}}}\ (\bibinfo  {publisher} {Springer},\ \bibinfo {address}
  {Berlin},\ \bibinfo {year} {2006})\ \bibinfo {note} {lecture Notes in
  Mathematics Vol. 1879. Ecole d'Et\'{e} de Probabilit\'{e}s de Saint--Flour
  XXXIV--2004}\BibitemShut {NoStop}%
\bibitem [{\citenamefont {Madras}\ and\ \citenamefont {Slade}(1993}]{MS93}%
  \BibitemOpen
  \bibfield  {author} {\bibinfo {author} {\bibfnamefont {N.}~\bibnamefont
  {Madras}}\ and\ \bibinfo {author} {\bibfnamefont {G.}~\bibnamefont {Slade}},\
  }\href@noop {} {\emph {\bibinfo {title} {The Self-Avoiding Walk}}}\ (\bibinfo
   {publisher} {Birkh{\"a}user},\ \bibinfo {address} {Boston},\ \bibinfo {year}
  {1993})\BibitemShut {NoStop}%
\bibitem [{\citenamefont {Kenna}(2004)}]{Kenna04}%
  \BibitemOpen
  \bibfield  {author} {\bibinfo {author} {\bibfnamefont {R.}~\bibnamefont
  {Kenna}},\ }\bibfield  {title} {\bibinfo {title} {Finite size scaling for
  ${O}({N})$ $\phi^4$-theory at the upper critical dimension},\ }\href@noop {}
  {\bibfield  {journal} {\bibinfo  {journal} {Nucl. Phys. B}\ }\textbf
  {\bibinfo {volume} {{\bf 691}}},\ \bibinfo {pages} {292} (\bibinfo {year}
  {2004})}\BibitemShut {NoStop}%
\bibitem [{\citenamefont {Ruiz-Lorenzo}(1998)}]{Ruiz98}%
  \BibitemOpen
  \bibfield  {author} {\bibinfo {author} {\bibfnamefont {J.}~\bibnamefont
  {Ruiz-Lorenzo}},\ }\bibfield  {title} {\bibinfo {title} {Logarithmic
  corrections for spin glasses, percolation and {Lee}--{Yang} singularities in
  six dimensions},\ }\href@noop {} {\bibfield  {journal} {\bibinfo  {journal}
  {J. Phys. A: Math. Gen.}\ }\textbf {\bibinfo {volume} {{\bf 31}}},\ \bibinfo
  {pages} {8773} (\bibinfo {year} {1998})}\BibitemShut {NoStop}%
\bibitem [{\citenamefont {Park}(2025{\natexlab{a}})}]{Park25b}%
  \BibitemOpen
  \bibfield  {author} {\bibinfo {author} {\bibfnamefont {J.}~\bibnamefont
  {Park}},\ }\bibfield  {title} {\bibinfo {title} {Torus scaling limits and the
  plateau of the critical weakly coupled $|\varphi|^4$ model in $d \ge 4$}}
  (\bibinfo {year} {2025}{\natexlab{a}}),\ \bibinfo {note} {preprint,
  \url{https://arxiv.org/abs/2511.06321}}\BibitemShut {NoStop}%
\bibitem [{\citenamefont {Heydenreich}\ and\ \citenamefont
  {Hofstad}(2007)}]{HHI07}%
  \BibitemOpen
  \bibfield  {author} {\bibinfo {author} {\bibfnamefont {M.}~\bibnamefont
  {Heydenreich}}\ and\ \bibinfo {author} {\bibfnamefont {R.~v.~d.}\
  \bibnamefont {Hofstad}},\ }\bibfield  {title} {\bibinfo {title} {Random graph
  asymptotics on high-dimensional tori},\ }\href@noop {} {\bibfield  {journal}
  {\bibinfo  {journal} {Commun. Math. Phys.}\ }\textbf {\bibinfo {volume} {{\bf
  270}}},\ \bibinfo {pages} {335} (\bibinfo {year} {2007})}\BibitemShut
  {NoStop}%
\bibitem [{\citenamefont {Hutchcroft}(2025)}]{Hutc25}%
  \BibitemOpen
  \bibfield  {author} {\bibinfo {author} {\bibfnamefont {T.}~\bibnamefont
  {Hutchcroft}},\ }\bibfield  {title} {\bibinfo {title} {Pointwise two-point
  function estimates and a non-pertubative proof of mean-field critical
  behaviour for long-range percolation}} (\bibinfo {year} {2025}),\ \bibinfo
  {note} {\url{https://doi.org/10.1007/s00440-025-01410-8}}\BibitemShut
  {NoStop}%
\bibitem [{\citenamefont {Duminil-Copin}\ and\ \citenamefont
  {Panis}(2025{\natexlab{a}})}]{DP25-Ising}%
  \BibitemOpen
  \bibfield  {author} {\bibinfo {author} {\bibfnamefont {H.}~\bibnamefont
  {Duminil-Copin}}\ and\ \bibinfo {author} {\bibfnamefont {R.}~\bibnamefont
  {Panis}},\ }\bibfield  {title} {\bibinfo {title} {New lower bounds for the
  (near) critical {Ising} and $\varphi^4$ models' two-point functions},\
  }\href@noop {} {\bibfield  {journal} {\bibinfo  {journal} {Commun. Math.
  Phys.}\ }\textbf {\bibinfo {volume} {{\bf 406}}},\ \bibinfo {pages} {56}
  (\bibinfo {year} {2025}{\natexlab{a}})}\BibitemShut {NoStop}%
\bibitem [{\citenamefont {Aizenman}(1982)}]{Aize82}%
  \BibitemOpen
  \bibfield  {author} {\bibinfo {author} {\bibfnamefont {M.}~\bibnamefont
  {Aizenman}},\ }\bibfield  {title} {\bibinfo {title} {Geometric analysis of
  $\varphi^4$ fields and {I}sing models, {Parts} {I} and {II}},\ }\href@noop {}
  {\bibfield  {journal} {\bibinfo  {journal} {Commun. Math. Phys.}\ }\textbf
  {\bibinfo {volume} {{\bf 86}}},\ \bibinfo {pages} {1} (\bibinfo {year}
  {1982})}\BibitemShut {NoStop}%
\bibitem [{\citenamefont {Eichelsbacher}\ and\ \citenamefont
  {L\"{o}we}(2010)}]{EL10}%
  \BibitemOpen
  \bibfield  {author} {\bibinfo {author} {\bibfnamefont {P.}~\bibnamefont
  {Eichelsbacher}}\ and\ \bibinfo {author} {\bibfnamefont {M.}~\bibnamefont
  {L\"{o}we}},\ }\bibfield  {title} {\bibinfo {title} {Stein's method for
  dependent random variables occurring in statistical mechanics},\ }\href@noop
  {} {\bibfield  {journal} {\bibinfo  {journal} {Electron. J. Probab.}\
  }\textbf {\bibinfo {volume} {{\bf 15}}},\ \bibinfo {pages} {962} (\bibinfo
  {year} {2010})}\BibitemShut {NoStop}%
\bibitem [{\citenamefont {Liu}\ and\ \citenamefont
  {Slade}(2025{\natexlab{b}})}]{LS25_profile}%
  \BibitemOpen
  \bibfield  {author} {\bibinfo {author} {\bibfnamefont {Y.}~\bibnamefont
  {Liu}}\ and\ \bibinfo {author} {\bibfnamefont {G.}~\bibnamefont {Slade}},\
  }\bibfield  {title} {\bibinfo {title} {Critical scaling profile for trees and
  connected subgraphs on the complete graph},\ }\href@noop {} {\bibfield
  {journal} {\bibinfo  {journal} {Canad. Math. Bull.}\ } (\bibinfo {year}
  {2025}{\natexlab{b}})},\ \bibinfo {note} {published online}\BibitemShut
  {NoStop}%
\bibitem [{\citenamefont {Gennes}(1972)}]{Genn72}%
  \BibitemOpen
  \bibfield  {author} {\bibinfo {author} {\bibfnamefont {P.~d.}\ \bibnamefont
  {Gennes}},\ }\bibfield  {title} {\bibinfo {title} {Exponents for the excluded
  volume problem as derived by the {Wilson} method},\ }\href@noop {} {\bibfield
   {journal} {\bibinfo  {journal} {Phys. Lett.}\ }\textbf {\bibinfo {volume}
  {{\bf A38}}},\ \bibinfo {pages} {339} (\bibinfo {year} {1972})}\BibitemShut
  {NoStop}%
\bibitem [{\citenamefont {Blanc-Renaudie}\ \emph {et~al.}(2024)\citenamefont
  {Blanc-Renaudie}, \citenamefont {Broutin},\ and\ \citenamefont
  {Nachmias}}]{BRBN26}%
  \BibitemOpen
  \bibfield  {author} {\bibinfo {author} {\bibfnamefont {A.}~\bibnamefont
  {Blanc-Renaudie}}, \bibinfo {author} {\bibfnamefont {N.}~\bibnamefont
  {Broutin}},\ and\ \bibinfo {author} {\bibfnamefont {A.}~\bibnamefont
  {Nachmias}},\ }\bibfield  {title} {\bibinfo {title} {The scaling limit of
  critical hypercube percolation}} (\bibinfo {year} {2024}),\ \bibinfo {note}
  {preprint, \url{https://arxiv.org/pdf/2401.16365}}\BibitemShut {NoStop}%
\bibitem [{\citenamefont {Blanc-Renaudie}\ and\ \citenamefont
  {Nachmias}(2025)}]{BRN26}%
  \BibitemOpen
  \bibfield  {author} {\bibinfo {author} {\bibfnamefont {A.}~\bibnamefont
  {Blanc-Renaudie}}\ and\ \bibinfo {author} {\bibfnamefont {A.}~\bibnamefont
  {Nachmias}},\ }\bibfield  {title} {\bibinfo {title} {Critical percolation on
  the discrete torus in high dimensions}} (\bibinfo {year} {2025}),\ \bibinfo
  {note} {preprint, \url{https://arxiv.org/pdf/2512.19672v1}}\BibitemShut
  {NoStop}%
\bibitem [{\citenamefont {Camia}\ \emph {et~al.}(2021)\citenamefont {Camia},
  \citenamefont {Jiang},\ and\ \citenamefont {Newman}}]{CJN21}%
  \BibitemOpen
  \bibfield  {author} {\bibinfo {author} {\bibfnamefont {F.}~\bibnamefont
  {Camia}}, \bibinfo {author} {\bibfnamefont {J.}~\bibnamefont {Jiang}},\ and\
  \bibinfo {author} {\bibfnamefont {C.}~\bibnamefont {Newman}},\ }\bibfield
  {title} {\bibinfo {title} {The effect of free boundary conditions on the
  {Ising} model in high dimensions},\ }\href@noop {} {\bibfield  {journal}
  {\bibinfo  {journal} {Probab. Theory Related Fields}\ }\textbf {\bibinfo
  {volume} {{\bf 181}}},\ \bibinfo {pages} {311} (\bibinfo {year}
  {2021})}\BibitemShut {NoStop}%
\bibitem [{\citenamefont {Chatterjee}\ and\ \citenamefont
  {Hanson}(2020)}]{CH20}%
  \BibitemOpen
  \bibfield  {author} {\bibinfo {author} {\bibfnamefont {S.}~\bibnamefont
  {Chatterjee}}\ and\ \bibinfo {author} {\bibfnamefont {J.}~\bibnamefont
  {Hanson}},\ }\bibfield  {title} {\bibinfo {title} {Restricted percolation
  critical exponents in high dimensions},\ }\href@noop {} {\bibfield  {journal}
  {\bibinfo  {journal} {Commun. Pure Appl. Math.}\ }\textbf {\bibinfo {volume}
  {{\bf 73}}},\ \bibinfo {pages} {2370} (\bibinfo {year} {2020})}\BibitemShut
  {NoStop}%
\bibitem [{\citenamefont {Hara}\ \emph {et~al.}(2003)\citenamefont {Hara},
  \citenamefont {Hofstad},\ and\ \citenamefont {Slade}}]{HHS03}%
  \BibitemOpen
  \bibfield  {author} {\bibinfo {author} {\bibfnamefont {T.}~\bibnamefont
  {Hara}}, \bibinfo {author} {\bibfnamefont {R.~v.~d.}\ \bibnamefont
  {Hofstad}},\ and\ \bibinfo {author} {\bibfnamefont {G.}~\bibnamefont
  {Slade}},\ }\bibfield  {title} {\bibinfo {title} {Critical two-point
  functions and the lace expansion for spread-out high-dimensional percolation
  and related models},\ }\href@noop {} {\bibfield  {journal} {\bibinfo
  {journal} {Ann. Probab.}\ }\textbf {\bibinfo {volume} {{\bf 31}}},\ \bibinfo
  {pages} {349} (\bibinfo {year} {2003})}\BibitemShut {NoStop}%
\bibitem [{\citenamefont {Duminil-Copin}\ and\ \citenamefont
  {Panis}(2025{\natexlab{b}})}]{DP25a}%
  \BibitemOpen
  \bibfield  {author} {\bibinfo {author} {\bibfnamefont {H.}~\bibnamefont
  {Duminil-Copin}}\ and\ \bibinfo {author} {\bibfnamefont {R.}~\bibnamefont
  {Panis}},\ }\bibfield  {title} {\bibinfo {title} {An alternative approach for
  the mean-field behaviour of weakly self-avoiding walks in dimensions $d>4$},\
  }\href@noop {} {\bibfield  {journal} {\bibinfo  {journal} {Probab. Theory
  Related Fields}\ } (\bibinfo {year} {2025}{\natexlab{b}})},\ \bibinfo
  {note} {\url{https://doi.org/10.1007/s00440-025-01415-3}}\BibitemShut
  {NoStop}%
\bibitem [{\citenamefont {Duminil-Copin}\ and\ \citenamefont
  {Panis}(2025{\natexlab{c}})}]{DP25b}%
  \BibitemOpen
  \bibfield  {author} {\bibinfo {author} {\bibfnamefont {H.}~\bibnamefont
  {Duminil-Copin}}\ and\ \bibinfo {author} {\bibfnamefont {R.}~\bibnamefont
  {Panis}},\ }\bibfield  {title} {\bibinfo {title} {An alternative approach for
  the mean-field behaviour of spread-out {Bernoulli} percolation in dimensons
  $d>6$},\ }\href@noop {} {\bibfield  {journal} {\bibinfo  {journal} {Probab.
  Theory Related Fields}\ } (\bibinfo {year} {2025}{\natexlab{c}})},\
  \bibinfo {note}
  {\url{https://doi.org/10.1007/s00440-025-01416-2}}\BibitemShut {NoStop}%
\bibitem [{\citenamefont {Kenna}\ \emph {et~al.}(2006)\citenamefont {Kenna},
  \citenamefont {Johnston},\ and\ \citenamefont {Janke}}]{KJJ06}%
  \BibitemOpen
  \bibfield  {author} {\bibinfo {author} {\bibfnamefont {R.}~\bibnamefont
  {Kenna}}, \bibinfo {author} {\bibfnamefont {D.}~\bibnamefont {Johnston}},\
  and\ \bibinfo {author} {\bibfnamefont {W.}~\bibnamefont {Janke}},\ }\bibfield
   {title} {\bibinfo {title} {Scaling relations for logarithmic corrections},\
  }\href@noop {} {\bibfield  {journal} {\bibinfo  {journal} {Phys. Rev. Lett.}\
  }\textbf {\bibinfo {volume} {{\bf 96}}},\ \bibinfo {pages} {115701} (\bibinfo
  {year} {2006})}\BibitemShut {NoStop}%
\bibitem [{\citenamefont {Bauerschmidt}\ \emph {et~al.}(2014)\citenamefont
  {Bauerschmidt}, \citenamefont {Brydges},\ and\ \citenamefont
  {Slade}}]{BBS-phi4-log}%
  \BibitemOpen
  \bibfield  {author} {\bibinfo {author} {\bibfnamefont {R.}~\bibnamefont
  {Bauerschmidt}}, \bibinfo {author} {\bibfnamefont {D.}~\bibnamefont
  {Brydges}},\ and\ \bibinfo {author} {\bibfnamefont {G.}~\bibnamefont
  {Slade}},\ }\bibfield  {title} {\bibinfo {title} {Scaling limits and critical
  behaviour of the $4$-dimensional $n$-component $|\varphi|^4$ spin model},\
  }\href@noop {} {\bibfield  {journal} {\bibinfo  {journal} {J. Stat. Phys}\
  }\textbf {\bibinfo {volume} {{\bf 157}}},\ \bibinfo {pages} {692} (\bibinfo
  {year} {2014})}\BibitemShut {NoStop}%
\bibitem [{\citenamefont {Bauerschmidt}\ \emph {et~al.}(2017)\citenamefont
  {Bauerschmidt}, \citenamefont {Slade}, \citenamefont {Tomberg},\ and\
  \citenamefont {Wallace}}]{BSTW-clp}%
  \BibitemOpen
  \bibfield  {author} {\bibinfo {author} {\bibfnamefont {R.}~\bibnamefont
  {Bauerschmidt}}, \bibinfo {author} {\bibfnamefont {G.}~\bibnamefont {Slade}},
  \bibinfo {author} {\bibfnamefont {A.}~\bibnamefont {Tomberg}},\ and\ \bibinfo
  {author} {\bibfnamefont {B.}~\bibnamefont {Wallace}},\ }\bibfield  {title}
  {\bibinfo {title} {Finite-order correlation length for 4-dimensional weakly
  self-avoiding walk and $|\varphi|^4$ spins},\ }\href@noop {} {\bibfield
  {journal} {\bibinfo  {journal} {Ann. Henri Poincar\'e}\ }\textbf {\bibinfo
  {volume} {{\bf 18}}},\ \bibinfo {pages} {375} (\bibinfo {year}
  {2017})}\BibitemShut {NoStop}%
\bibitem [{\citenamefont {Park}(2025{\natexlab{b}})}]{Park25a}%
  \BibitemOpen
  \bibfield  {author} {\bibinfo {author} {\bibfnamefont {J.}~\bibnamefont
  {Park}},\ }\bibfield  {title} {\bibinfo {title} {A renormalisation group map
  for short- and long-ranged weakly coupled $|\varphi|^4$ models in $d \ge 4$
  at and above the critical point}} (\bibinfo {year} {2025}{\natexlab{b}}),\
  \bibinfo {note} {preprint,
  \url{https://arxiv.org/abs/2511.03495}}\BibitemShut {NoStop}%
\end{thebibliography}

%

\appendix

\section{Elementary lemmas}
\label{app:conv}

This appendix contains two lemmas.
The first is a restatement of Lemma~\ref{lem:periodic_sum}, with proof. The proof
uses the simple observation that for $x\in \T_R^d$
(identified with a point in $\Z^d$ with $|x| \le R/2$) and for
nonzero $u\in \Zd$, the points $x+Ru$ and $u$ are simply related by
\begin{equation} \label{eq:xubd}
\begin{aligned}
|x+Ru|
&\ge R|u| - \frac R 2
\ge \frac R 2 |u|,
\\
|x+Ru|
&\le R | u| + \frac R 2
\le \frac 3 2 R | u|.
\end{aligned}
\end{equation}

\begin{lemma} \label{lem:periodic_sum-app}
Let $d \ge 1$, $R \ge 1$,
$a >0$ and $\xi >0$.
Suppose $g:[0,\infty) \to [0,\infty)$ satisfies $g(s) \le (1+s)^{-(a+\eps)}$ for some $\eps > 0$.
Then there is a constant $C = C(d,a,\eps)>0$ such that
\begin{equation*}
\sum_{u \in\Z^d : u \neq 0}
\frac{1}{|x + R u|^{d-a}} g \bigg( \frac{ |x+Ru| } \xi \bigg)
\le C \frac{ \xi^a }{R^d}
\end{equation*}
for all $x\in \Zd$ with $| x| \le R/2$.
\end{lemma}

\begin{proof}
We may assume that $g(s) = (1+s)^{-(a+\eps)}$.
By \eqref{eq:xubd} and the monotonicity of $g$
it is enough to obtain an upper bound of order $\xi^a/R^d$ for
\begin{align}
&\sum_{u \in\Z^d : u \neq 0}
\frac{1}{ ( \frac R 2 | u| )^{d-a}}
 g \bigg( \frac R {2\xi} | u| \bigg) 	\nl
&\qquad \lesssim
\frac{ 1}{ R^{d-a} }
\sum_{N=1}^\infty N^{a-1}
 g \bigg( \frac R {2\xi} N \bigg) .
\end{align}
If $0 < a \le 1$,
the summand is decreasing in $N$, so we can bound the sum by the integral
\begin{align}
\int_0^\infty t^{a-1}  g \bigg( \frac R {2\xi} t \bigg) \D t
&= \bigg( \frac{2\xi}R \bigg)^{a}
	\int_0^\infty y^{a-1}  g (y) \D y 	\nl
&= C_{a,\eps} \frac{ \xi^a }{ R^a },
\end{align}
using $a > 0$ and $\eps > 0$.

If $a > 1$, by the definition of $g$ it suffices to bound
\begin{equation}
\sum_{N=1}^{ \lfloor 2\xi / R \rfloor} N^{a-1}
+
\sum_{N= \max\{\lceil 2\xi / R \rceil , 1\}}^\infty N^{a-1} \bigg( \frac R {2\xi} N \bigg)^{-(a+\eps)} .
\end{equation}
The first term is $O(\xi^a / R^a)$ since $a> 1$.
The second term is a multiple of
\begin{equation}
\label{eq:age1}
    \frac{\xi^{a+\varepsilon}}{R^{a+\varepsilon}}
    \sum_{N= \max\{\lceil 2\xi / R \rceil , 1\}}^\infty N^{-1-\eps}.
\end{equation}
If the maximum in the lower summation limit is $1$ then $\xi/R \le 1$ and \eqref{eq:age1}
is of order $(\xi/R)^{a+\varepsilon} \le (\xi/R)^a$.
If instead the maximum is $\lceil 2\xi / R \rceil$ then \eqref{eq:age1} is bounded
by a multiple of $(\xi/R)^{a+\eps}(\xi / R)^{-\varepsilon} = (\xi/R)^{a}$.
This completes the proof.
\end{proof}

The second lemma is an elementary convolution estimate,
which is applied in Section~\ref{sec:SAW} to verify Hypothesis~\ref{ass:Gbar} for the spread-out LR SAW.

\begin{lemma} \label{lem:G2}
Let $\alpha > 0$, $d > 2\alpha$, and $\xi \in [1, \infty]$.
Suppose $G: \Zd \to [0,\infty)$ obeys
\begin{equation}
G(x) \lesssim \frac 1 { \nnnorm x^{d-\alpha} } \frac 1 {(1 + \abs x / \xi)^{2\alpha}} .
\end{equation}
Then there is a constant independent of $\xi$ such that
\begin{equation}
(G*G)(x) \lesssim \frac 1 { \nnnorm x^{d-2\alpha} } \frac 1 {(1 + \abs x / \xi)^{3\alpha}} .
\end{equation}
\end{lemma}

\begin{proof}
By \cite[Proposition~1.7]{HHS03}, we have
$(G*G)(x) \lesssim  \nnnorm x^{-(d-2\alpha)} $ for all $x$,
so it suffices to improve the bound when $\abs x \ge 2\xi$ to
\begin{equation} \label{eq:LR_decay_claim}
(G*G)(x) \lesssim \frac {\xi^{3\alpha}} { \abs x^{d+\alpha} } .
\end{equation}
In the sum $(G*G)(x) = \sum_{y\in \Zd} G(y) G(x-y)$,
either $\abs{x-y}$ or $\abs y$ must be at least $\frac 1 2 \abs x$.
In the first case, we use the assumed bound on $G(x-y)$ to get
\begin{equation}
\sum_{y\in \Zd : \abs{x-y} \ge \frac 1 2 \abs x}
G(y) G(x-y)
\lesssim  \frac {\xi^{2\alpha}} { ( \frac 1 2 \abs x )^{d+\alpha} }
	\sum_{y\in \Zd} G(y) ,
\end{equation}
and we use the hypothesis again to see that
\begin{equation}
\sum_{y\in \Zd} G(y)
\lesssim \sum_{\abs y \le \xi} \frac 1 { \nnnorm y^{d-\alpha} }
	+ \sum_{\abs y \ge \xi} \frac{ \xi^{2\alpha} }{ \abs y^{d+\alpha} }
\lesssim \xi^\alpha .
\end{equation}
This gives an upper bound
$\xi^{3\alpha} / \abs x^{d+\alpha}$ for the $\abs{x-y} \ge \frac 1 2 \abs x$ part of the convolution.
The $\abs{y} \ge \frac 1 2 \abs x$ part is analogous, with the decay
coming from $G(y)$ instead of $G(x-y)$.
This gives \eqref{eq:LR_decay_claim} and completes the proof.
\end{proof}

\section{Lower bound of Hypothesis~\ref{ass:Gbar} for short-range models}
\label{app:LB}

Hypothesis~\ref{ass:G} is
proved for various models in \cite{Slad23_wsaw,DP25a,DP25b,HMS23,DP25-Ising,LS25a,Liu25EJP,Hutc25,Liu25}.
In this appendix, we briefly sketch the proof
of the lower bound of \eqref{eq:GGbar} in Hypothesis~\ref{ass:Gbar},
for short-range self-avoiding walk,
branched polymers, percolation, and the Ising model.  References are given where full details can be found.

We use the fact that Hypothesis~\ref{ass:G} has been established in previous works with $g = g_{\mathrm{SR}}$.
Our goal is to prove that there are constants $c_0,s_2$ such that,
when $\xi(\beta) \ge s_2R$, we have
\begin{equation} \label{eq:GGlb}
\Gbar_{R,\beta}(x) - G_{R,\beta}(x)
\le c_0 \frac{\chi(\beta)^{\dc/2}}{V} \Gamma_{R,\beta}(x) .
\end{equation}

The model dependence in this bound arises via the diagrams for $\Gbar_{R,\beta}(x)-G_{R,\beta}(x)$ depicted in Figure~\ref{figure:Zd_diagrams}, which  project onto the torus diagrams in Figure~\ref{figure:diagrams other models}.
We estimate the diagrams using the torus convolution
\begin{equation}
(f\star g)(x) = \sum_{y\in \T_R^d} f(x-y)g(y)
\end{equation}
for functions $f,g:\T_R^d \to \R$.
From the near-critical upper bound \eqref{eq:Gdecay}, a convolution estimate, and Lemma~\ref{lem:periodic_sum}, the $m$-fold convolution of $\Gamma_{R,\beta}$ with itself satisfies, for $d>2m$,
\begin{equation}
    \Gbar_{R,\beta}^{\star m}(x)
    	\lesssim \frac{1}{\nnnorm x^{d-2m}} + \frac{\chi^m}{V}
    .
\end{equation}
Together with $|x|^2 \le (\frac R 2)^2 \lesssim \xi^2 \asymp \chi$
and the lower bound of $\Gamma_{R,\beta}(x)$ from Proposition~\ref{prop:Gamma},
this implies the useful inequality
\begin{equation} \label{eq:fstar}
    \Gbar_{R,\beta}^{\star m}(x)
    	\lesssim \chi^{m-1} \Gbar_{R,\beta} (x) .
\end{equation}

\begin{figure}[h]
\includegraphics[scale = 0.6]{saw}
\quad
\includegraphics[scale = 0.6]{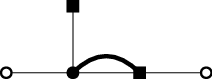}
\quad
\includegraphics[scale = 0.6]{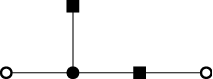}

\vspace{3mm}
\includegraphics[scale = 0.6]{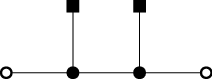}
\qquad
\includegraphics[scale = 0.6]{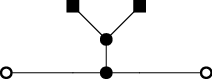}
\caption{$\Z^d$ configurations contributing to $\Gbar_{R,\beta}(x)-G_{R,\beta}(x)$.
First line, left to right: SAW, Ising, percolation.
Second line: two topologies for BP.
Thin lines represent $G_\beta$; the bold line for Ising represents $\Gbar_{R,\beta}$.
Hollow vertices are $0$ and $x' = x + Ru$,
filled vertices are summed over $\mathbb Z^d$,
and box vertices are summed over torus equivalent, distinct, unordered pairs in $\Z^d$.
}
\label{figure:Zd_diagrams}
\end{figure}

\begin{figure}[h]
\includegraphics[scale = 0.6]{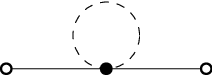}
\quad
\includegraphics[scale = 0.6]{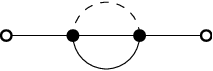}
\quad
\includegraphics[scale = 0.6]{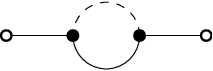}

\vspace{3mm}
\includegraphics[scale = 0.6]{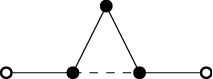}
\quad
\includegraphics[scale = 0.6]{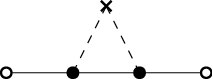}
\quad
\includegraphics[scale = 0.6]{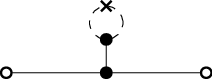}
\caption{Torus configurations contributing to $\Gbar_{R,\beta}(x)-G_{R,\beta}(x)$.
First line, left to right: SAW, Ising, percolation.
Second line:  three diagrams for BP.  Solid lines represent $\Gbar_{R,\beta}$,
dashed lines represent $\chi/V$,
two dashed lines connected by $\times$ represent $\chi^2/V$.
Hollow vertices are $0,x$ and filled vertices
are summed over the torus.}
\label{figure:diagrams other models}
\end{figure}

\subsection{Self-avoiding walk}

For SR SAW, the inequality \eqref{eq:GGlb} follows exactly as in
the derivation of \eqref{eq:GGbar} for LR SAW in Section~\ref{sec:SAW},
with $\alpha$ set equal to $2$.
The SAW diagram of Figure~\ref{figure:Zd_diagrams} was
encountered in Section~\ref{sec:SAW}.  It projects to the torus
SAW diagram in Figure~\ref{figure:diagrams other models}.

\subsection{Branched polymers}

Similarly to SAW, the difference $\Gbar_{R,\beta}(x)-G_{R,\beta}(x)$
for lattice trees arises from $\Z^d$ trees that do not project to a torus tree.
These trees must contain two distinct points that project to the same torus point.
This can happen with two topologies, as depicted in Figure~\ref{figure:Zd_diagrams},
depending on whether the unique path joining the distinct points
intersects the path joining $0$ and $x'$.
Analogous considerations to those for SAW
lead to an upper bound
by the three BP diagrams in Figure~\ref{figure:diagrams other models} \cite{LS25a}.
Using \eqref{eq:fstar},
these diagrams are bounded, for $m=4,3,2$ respectively, by
\begin{align}
    \frac{\chi^{5-m}}{V} \Gbar_{R,\beta}^{\star m} (x)
		& \lesssim
    \frac{\chi^{5-m}}{V}
    \chi^{m-1} \Gbar_{R,\beta}(x)
    =
    \frac{\chi^{4}}{V}
    \Gbar_{R,\beta}(x) .
\end{align}
The desired inequality \eqref{eq:GGlb} then follows since $\dc = 8$.

For lattice animals, minimal additional care is required to
specify  the unwrapping operation.
Once this is done, the rest is the same as for lattice trees \cite{LS25a}.

\subsection{Percolation}

A coupling of percolation on the torus and the infinite lattice $\Z^d$ is
discussed in detail in \cite{HHI07,HMS23}.  It involves an \emph{exploration process}
which provides the
unwrapping of a torus percolation cluster
to the infinite lattice.  From this, the upper
bound $G_{R,\beta}(x) \le \Gbar_{R,\beta}(x)$ follows directly.
It also leads to an upper bound on the difference
$\Gbar_{R,\beta}(x) - G_{R,\beta}(x)$ in which there is a connection in $\Z^d$ from
$0$ to a point $x'$ which projects to $\pi(x')=x$, passing through some point $y$, with $0$ also connected
to a point $y' \neq y$ with $\pi(y')= \pi(y)$.
A generic such configuration is depicted in
Figure~\ref{figure:Zd_diagrams}, and its representation as a
torus diagram is depicted in Figure~\ref{figure:diagrams other models}.
By \eqref{eq:fstar} with $m=3$, the torus diagram is bounded by
\begin{align}
\frac{ \chi } V \Gbar_{R,\beta}^{\star 3}(x)
	\lesssim \frac{\chi^3}{V} \Gbar_{R,\beta}(x).
\end{align}
This gives \eqref{eq:GGlb} because $\dc = 6$.

\subsection{Ising model}

The unwrapping procedure for the Ising model uses the random current representation,
which provides a percolation-like geometric representation for Ising correlation
functions.
The advantages of this geometric representation
were first championed in \cite{Aize82}.
A \emph{current} configuration $\bn$ is an assignment of a non-negative integer $\bn_b$ to each bond
$b$.  A lattice site $x$  is a \emph{source} if the sum of $\bn_b$ over bonds incident to $x$
is odd.  The set of sources of $\bn$ is denoted $\partial\bn$.
Let $w_\beta(\bn) = \prod_{b} \beta^{\bn_b} / \bn_b !$.
The random current
representation for the Ising two-point function is
\begin{align}
    G_\beta (x) =
    \frac{\sum_{\bn : \partial \bn = \{0,x\}}w_\beta(\bn)}
        {\sum_{\bn : \partial \bn = \varnothing}w_\beta(\bn)}.
\end{align}

Via this representation, a coupling between the Ising model on the torus
and on the infinite lattice is established in \cite{LPS25-Ising}.
The bound $G_{R,\beta}(x) \le \Gbar_{R,\beta}(x)$ is proved using this coupling.
The coupling also shows that $\Gbar_{R,\beta}(x)-G_{R,\beta}(x)$ is dominated by a
configuration diagrammatically depicted in Figure~\ref{figure:Zd_diagrams},
which projects to the torus diagram of
Figure~\ref{figure:diagrams other models}.  The latter diagram
is bounded using an extension of \eqref{eq:fstar} which implies that
\begin{align}
    \frac{\chi}{V}(\Gbar_{R,\beta} \star \Gbar_{R,\beta}^2 \star \Gbar_{R,\beta})(x)
    \lesssim
    \frac{\chi^2}{V} \Gbar_{R,\beta}(x).
\end{align}
Details are given in \cite{LPS25-Ising}.

\section{The exponents $\q$ and $\hat\q$}
\label{app:coppa}

In much of the literature on high-dimensional finite-size scaling
under PBC, an exponent
$\q=d/\dca$ (coppa) is computed via a connection with Lee--Yang zeros, e.g.,
\cite{F-SBKW16,BEHK22}.  The exponent $\q$
gives the divergence of a correlation length of order $R^{\qq}$ which exceeds the
system size $R$.  At the critical dimension $d=\dca$, there is a logarithmic
correction $R^{\qq}(\log R)^{\hat{\qq}}$, e.g., \cite{KJJ06}.
Both $\q$ and $\hat\q$ occur in various
scaling relations. Originally, $\q,\hat\q$ were called $q,\hat q$.

In our work, a connection with Lee--Yang zeros is not made.  On the other hand,
the correlation length $\xi$ of the infinite system plays a prominent role.
Its value at the point $\beta_*$ in the critical window obeys
\begin{equation}
    \xi(\beta_*) \asymp (V^{-\frac{2}{\gamma \dc}})^{-\nu}
    =
    R^{\frac{2\nu}{\gamma}\frac{d}{\dc}}
    =
    R^{\frac{2}{\alpha}\frac{d}{\dc}}
    =
    R^{\frac{d}{\dca}}.
\end{equation}
The exponent on the right-hand side agrees with and provides an alternate
interpretation for $\q$.

For the SR $n$-component $|\varphi|^4$ model in the critical dimension
$d=\dc=4$, it is well-known and rigorously proved
\cite{BBS-phi4-log,BSTW-clp} that $\gamma=2\nu=1$ with logarithmic
exponents $\hat\gamma = 2\hat\nu=
\frac{n+2}{n+8}$.
Also, the specific heat has exponent $\alpha=0$ for $n \ge 1$,
with
$\hat\alpha = \frac{4-n}{n+8}$ for $n=1,2,3$, whereas for $n=4$
the specific heat diverges
as $\log\log t^{-1}$ and for $n > 4$
it does not diverge at all \cite{BBS-phi4-log}.
In particular, $\hat\alpha$ is never strictly negative.
Let $\hat\theta = \frac{4-n}{2(n+8)}=\frac 12 -\hat\gamma$, for all $n\ge 1$ including $n \ge 4$.
It is argued in \cite[(3.6), (4.3)]{Kenna04} that under PBC the finite-volume
susceptibility has a log correction
$\chi_R \asymp V^{1/2}(\log V)^{1/2}$ and that the window width is
$V^{-1/2}(\log V)^{-\hat\theta}$.
It was proved rigorously that $\chi_R(\beta_c)\asymp V^{1/2}(\log V)^{1/2}$ in \cite{Park25b} using a rigorous RG analysis \cite{Park25a}.
For the hierarchical model, precise asymptotics for both the torus
susceptibility and window width were proved in \cite{MPS23}.

At the high-temperature edge of the window, the
infinite-volume correlation length therefore obeys
\begin{equation}
    \xi(\beta_*) \asymp V^{1/4}(\log V)^{\frac{\hat\theta}{2} + \hat\nu}
    =
    V^{1/4}(\log V)^{1/4}.
\end{equation}
If we interpret the right-hand side as $R^{\qq}(\log R)^{\hat\qq}$,
then we find that $\q=1$ and $\hat\q = \frac 14$ for all $n \ge 1$ (in agreement with
 \cite{Kenn12}), and that
\begin{equation}
\label{eq:qhat}
    2\hat\theta = 4(\hat\q -\hat\nu) = 1-4\hat\nu.
\end{equation}
This agrees with the scaling relation $\hat\alpha = 4(\hat\q-\hat\nu)$
from \cite[(1.39)]{Kenn12} when $n=1,2,3$.
For $n=4$, the exponent $\hat\alpha$ is zero in a logarithmic sense,  whereas
both sides of \eqref{eq:qhat} are exactly zero.  For $n>4$, the exponent
$\hat\alpha$ is ill-defined since the specific heat does not diverge.  Nevertheless,
\eqref{eq:qhat} remains valid for all $n\ge 1$, including $n >4$.  In this sense,
the scaling relation \eqref{eq:qhat} has wider validity than $\hat\alpha = 4(\hat\q-\hat\nu)$.

\section{Renormalisation group}
\label{app:RG}

In this appendix, we do not work in a mathematically rigorous manner.
Instead, we argue as in \cite{BZ85,Zinn21} to
compute the profile
$f_n$ for the (short- or long-range) $n$-component lattice $|\varphi|^4$ model
in dimensions $d>\dca=2\alpha$ under PBC.
The hypothesis of universality suggests
that the same scaling and profile found in this appendix
should apply to other $O(n)$ models above the upper critical
dimension, including Ising ($n=1$) and $XY$ ($n=2$) models.

The $|\varphi|^4$ model has spins $\varphi_x \in \R^n$ and
Hamiltonian
\begin{equation}
\label{eq:phi4H}
	H=
    \frac{1}{2} (\varphi,   (-\Delta)^{\alpha / 2}  \varphi)
    + \sum_{x}  (\frac{g}{4}  |\varphi_x|^4 + \frac{\nu}{2} |\varphi_x|^2
	+ \vec{h} \cdot  \varphi_x)
	,
\end{equation}
where $\vec{h}=(h,\ldots,h)\in \R^n$ is a constant external field and $\Delta$
is the discrete Laplacian.
We consider both SR (set $\alpha=2$ in all formulas)
and LR ($0<\alpha<\min\{2,d\}$) models
in dimensions $d> \dca=2\alpha$.
An observable $F$ has expectation
\begin{equation}
    \langle F \rangle = \frac{1}{Z} \int F(\varphi) e^{-H(\varphi)}
    D\varphi.
\end{equation}
The two-point function and  susceptibility are
\begin{equation}
    G_\nu(x) =  \frac 1n \langle \varphi_0\cdot \varphi_x\rangle_{h=0},
    \quad \chi(\nu) = \sum_{x\in\Z^d} G_\nu(x).
\end{equation}
For $d \ge 2\alpha$,
there is an infinite-volume critical value
$\nuc<0$ such that the Wilson renormalisation
group (RG) flow started from
$\nu=\nuc$  converges to a stable Gaussian fixed point.

For PBC, a theory of finite-size scaling based on the
RG is developed in \cite{BZ85} and \cite[Section~32.3]{Zinn21}.
That work considered short-range interactions, but it extends
mutatis mutandis to long-range, as we present here.
We consider now a volume $V=R^d$ with PBC in dimensions $d>2\alpha$.

Under the change of scale $x \mapsto \scale^{-1} x$,
the field scales as
$\varphi \mapsto \scale^{-\frac{d-\alpha}{2}}\varphi$.
Let $\nu=\nuc + t$.
The model remains near the Gaussian fixed point
as long as $\scale^{\alpha}|t|$ remains bounded.
In this regime, briefly put, the free energy is renormalised to linear order as
\begin{align}
	f(t , g,h) =  \scale^{-d} f( \scale^{y_t} t  ,  \scale^{y_i} g ,  \scale^{y_h} h  )
	+ \delta f,
\end{align}
where
$y_t = \alpha$, $y_i = - (d - 2 \alpha)$, $y_h = \frac{d+\alpha}{2}$,
and $\delta f$ is an inhomogeneous term which does not play a role in the present
calculation.
Near the fixed point,
$t (\scale) \sim a_d \scale^{y_t} t$
and $g (\scale) \sim \scale^{y_i}  \bar{g}$ for some constants
$a_d$ and $\bar{g}$.
(We write $f_1 \sim f_2$ for $\lim f_1/f_2 =1$.)

The susceptibility on the torus $\mathbb{T}^d_R$
can be computed via $\chi_R = n^{-1} R^{-d} \frac{\partial^2}{\partial h^2}  \log Z |_{h=0} $.
On the torus,
the Gaussian field with covariance $((-\Delta)^{\alpha /2} )^{-1}$ can be decomposed into two independent Gaussian fields $\zeta + \psi$, where $\psi$ is the
constant zero mode
and $\zeta$ is the rest of the spin field.
Integrating out $\zeta$ corresponds to performing the RG flow to scale $\scale =R$.
Then $\chi_R$ is given by the zero-mode integral
\begin{equation}
\label{eq:zero-mode}
	\chi_R
    \sim
    R^{\alpha}
    \frac{ \int_{\R^n} |\psi|^2  e^{- U_R(\psi)} \mathrm d \psi }
    {n \int_{\R^n} e^{-  U_R(\psi)}   \mathrm  d \psi}
\end{equation}
with
\begin{equation}
\label{eq:URdefi}
	 U_R(\psi)
    =
    \frac{1}{4} R^{-(d - 2 \alpha) } \bar{g} |\psi|^4
    + \frac{1}{2}a_d R^\alpha  t |\psi|^2 .
\end{equation}
The importance of the zero mode for finite-size scaling under PBC
has been stressed in \cite{WY14} from a different perspective.

After the change of variables $\psi \mapsto \bar{g}^{-1/4} R^{(d - 2 \alpha)/4} \psi$,
and for $\nu = \nuc+ sa_d^{-1} \bar{g}^{1/2} R^{-d/2}$
with $s \in (-\infty,\infty)$ (window scale $V^{-1/2}$),
this becomes
\begin{equation}
\label{eq:chi-profile}
    \chi_R
    \sim
     \bar{g}^{-1/2} V^{1/2} f_n(s),
\end{equation}
with the non-Gaussian profile $f_n$ given by
\begin{equation}
    f_n(s)
    =
    \frac{\int_{\R^n} |x|^2  e^{-\frac 14 x^4  -   \frac 12 s x^2}\mathrm dx}
    {\int_{\R^n}   e^{-\frac 14 x^4  -   \frac 12 s x^2}\mathrm dx} .
\end{equation}
After conversion to polar coordinates, this agrees with the formula for
$f_n$ in \eqref{eq:fndef}.
With the same value of $\nu$, the torus plateau is given by
\begin{equation}
\label{eq:RGplateau}
    \frac{\chi_R}{V}
    \sim
    \frac{f_n(s)}{\bar g^{1/2}V^{1/2}}
    ,
\end{equation}
which dominates the Gaussian contribution $|x|^{-(d-\alpha)}$ to $G_{R}(x)$
when $|x| \gtrsim R^{\frac{d}{2(d-\alpha)}}$; note that the exponent
can be written as $(2-\dca/d)\inv$ so is less than $1$.
If instead $t=\nu-\nuc \gg R^{-d/2}$ then a crossover
occurs to Gaussian behaviour
when $t$ is of order $R^{-\alpha}$.  In particular, for such $t$ the
susceptibility scales as $R^\alpha$, not as $R^{d/2}$.
The crossover is discussed in detail for the hierarchical model in \cite{MPS23}.

The above identifies the critical window with scale $V^{-1/2}$,
the susceptibility scale $V^{1/2}$, the plateau scale $V^{-1/2}$,
and the universal
profile $f_n(s)$, as in Tables~\ref{table:fss} and~\ref{table:profile}.

Universal ratios can be computed similarly.  E.g., with the same
$\nu=\nuc+s a_d^{-1} \bar g^{1/2}V^{-1/2}$,
the moments of the total spin $\Phi_R = \sum_{x\in \T_R^d}\varphi_{x}$ obey,
as $R \to \infty$,
\begin{equation}
\label{eq:Phiratio}
    \frac{\langle |\Phi_R|^4\rangle}{\langle |\Phi_R|^2 \rangle^2}
    \to
    \frac{\int_{\R^n} |x|^4  e^{-\frac 14 x^4  -   \frac 12 s x^2}\mathrm dx}
    {(\int_{\R^n} |x|^2  e^{-\frac 14 x^4  -   \frac 12 s x^2}\mathrm dx)^2} .
\end{equation}
The dimensionless ratio \eqref{eq:Phiratio} appears in the \emph{Binder cumulant} and
the \emph{renormalised coupling constant}.
At $\nuc$ ($s=0$), the right-hand side
of \eqref{eq:Phiratio} takes the value
\begin{equation}
           \frac n4
    \left(
    \frac{\Gamma(\frac{n}{4})}
    { \Gamma(\frac{n+2}{4})}
    \right)^{2}.
\end{equation}
\end{document}